\documentclass{article}

\usepackage{arxiv}
\usepackage{url}            %
\usepackage{booktabs}       %
\usepackage{amsfonts}       %
\usepackage{nicefrac}       %
\usepackage{amsmath}
\usepackage{amssymb}
\usepackage{amsthm}
\usepackage{bbm}
\usepackage{natbib}
\usepackage{mathtools}
\usepackage{algorithm}
\usepackage{algorithmic}

\usepackage[inline]{enumitem}

\usepackage{thmtools} 
\usepackage{thm-restate}

\newtheorem{lemma}{Lemma}
\newtheorem{assumption}{Assumption}

\newtheorem{proposition}{Proposition}
\newtheorem{definition}{Definition}

\newtheorem{remark}{Remark}

 \def\independenT#1#2{\mathrel{\rlap{$#1#2$}\mkern2mu{#1#2}}}
\newcommand\indep{\protect\mathpalette{\protect\independenT}{\perp}}

\newcommand\posscite[1]{\citeauthor{#1}'s (\citeyear{#1})}

\newcommand{\Cov}{\mathrm{Cov}}
\newcommand{\II}{\boldsymbol{I}}
\newcommand{\XX}{\boldsymbol{X}}
\newcommand{\BB}{\boldsymbol{B}}
\newcommand{\MM}{\boldsymbol{M}}
\newcommand{\QQ}{\boldsymbol{Q}}
\newcommand{\aaa}{\boldsymbol{a}}
\newcommand{\bb}{\boldsymbol{b}}
\newcommand{\ee}{\boldsymbol{e}}
\newcommand{\vv}{\boldsymbol{v}}
\newcommand{\xx}{\boldsymbol{x}}
\newcommand{\sss}{\boldsymbol{s}}
\newcommand{\ww}{\boldsymbol{w}}
\newcommand{\uu}{\boldsymbol{u}}
\newcommand{\zz}{\boldsymbol{z}}
\newcommand{\eeta}{\boldsymbol{\eta}}
\newcommand{\Pv}{\boldsymbol{P}_v}
\newcommand{\Px}{\boldsymbol{P}_x}
\newcommand{\Qv}{\boldsymbol{Q}_v}
\newcommand{\Qx}{\boldsymbol{Q}_x}
\newcommand{\PP}{\boldsymbol{P}}

\newcommand{\mmu}{\boldsymbol{\mu}}
\newcommand{\br}[1]{\left({#1} \right)}
\newcommand{\brs}[1]{\left[{#1} \right]}
\newcommand{\E}{\mathrm{E}}
\newcommand{\Var}{\mathrm{Var}}

\title{Online Balanced Experimental Design}
\author{
  David Arbour* \\
  Adobe Research\\
  345 Park Avenue\\
  San Jose, CA \\
  \texttt{darbour26@gmail.com} \\
   \And
 Drew Dimmery* \\
  Data Science @ University of Vienna\\ Vienna, AT\\
  \texttt{pddimmery@gmail.com} \\
  \And
  Tung Mai\\
  Adobe Research\\
  345 Park Avenue\\
  San Jose, CA \\
  \texttt{tumai@adobe.com} \\
  \And
  Anup Rao\\
  Adobe Research\\
  345 Park Avenue\\
  San Jose, CA \\
  \texttt{anuprao@adobe.com} \\
}
\begin{document}

\maketitle

\begin{abstract}
    We consider the experimental design problem in an online environment, an important practical task for reducing the variance of estimates in randomized experiments which allows for greater precision, and in turn, improved decision making. 
  In this work, we present algorithms that build on recent advances in online discrepancy minimization which accommodate both arbitrary treatment probabilities and multiple treatments. 
  The proposed algorithms are computational efficient, minimize covariate imbalance, and include randomization which enables robustness to misspecification.
  We provide worst case bounds on the expected mean squared error of the causal estimate and show that the proposed estimator is no worse than an implicit ridge regression, which are within a logarithmic factor of the best known results for offline experimental design.
  We conclude with a detailed simulation study showing favorable results relative to complete randomization as well as to offline methods for experimental design with time complexities exceeding our algorithm.
\end{abstract}

\section{Introduction}
Randomized experimentation is a fundamental tool for obtaining counterfactual estimates. 
The efficacy of randomization comes from a very simple intuition--by randomly assigning treatment status dependence between observed (and unobserved) treatment and pre-treatment covariates necessarily tends to zero as a function of the number of units.
In the context of experimentation, this independence condition on observed covariates, commonly known as balance~\citep{imai2008misunderstandings, imai2014covariate}, reduces the variance of estimates of the average treatment effect~\citep{greevy2004optimal, higgins2016improving, kallus2018optimal, li2018asymptotic, harshaw2020balancing}.
Under the appropriate conditions such corrections can result in large increases in effective sample size, allowing for the detection of the small effects which are commonplace in many large scale studies and industrial applications~\citep{dimmery2019shrinkage, azevedo2020b}.
These contexts rely heavily on experimentation for decision-making, so reduced variance directly translates into more reliable decisions~\citep{kohavi2012trustworthy}.

Traditional experimental design like blocking~\citep{greevy2004optimal, higgins2016improving} or even the novel Gram-Schmidt Walk design~\citep{harshaw2020balancing} require more than one pass over the sample and their sample complexity is greater than $\mathcal{O}(n)$. 
Even algorithms which admit sequential assignment such as \citet{moore2017blocking} suffer from the fact that the algorithm is not linear time and, thus, respondents late in the experiment may take substantially longer to receive a treatment assignment~\citep{cavaille2018implementing}. 
Our work is motivated by this setting to provide a linear-time, single-pass (i.e. sequential) algorithm for balancing experimental design. 
Our focus is on linear measures of balance (often of particular interest to applied researchers). 
This provides a new avenue through which experimenters can ensure that their experiments optimize the information they gain from costly samples.

We start by presenting four desiderata for effective, practical online experimental design.
First, a method must be computationally efficient.
In a review of existing methods for online assignment in the case of survey experiments, \citet{cavaille2018implementing} finds that existing methods become slow to unusable as increasingly more respondents are included in a study.
This resulting speed is fundamentally incompatible with effective administration of an experiment.
In short, high latency will cause disproportionately large dropoff in an experiment, which may completely nullify the gains from using more sophisticated experimental design.
Any algorithm with greater than linear time complexity exacerbate this problem: higher latency for later subjects than earlier subjects will tend to cause non-random sample attrition, as later subjects (who may be different than earlier respondents) will be more likely to drop out.

Second, an experimental design must reduce covariate imbalance to be effective.
This is the entire justification for using methods more sophisticated than Bernoulli randomization, so if the design is not able to improve on balance, then there will be no subsequent reduction in variance and therefore no compelling reason to use it.

Third, the design must incorporate randomization.
\citet{harshaw2020balancing} provides an extensive discussion on the inherent tradeoff between robustness and balance within experimental design.
A design which solely optimizes for balance will tend to operate on a knife-edge of accidental bias~\citep{efron1971forcing}: the potential bias from an adversarially chosen confounder.
With higher accidental bias than Bernoulli randomization, if units do not arrive precisely i.i.d., then the entire design may be compromised.
Given that experimental settings are prized specifically for their unbiasedness, this could completely undermine any gains from improved balance.
Strong theoretical guarantees on robustness are extremely important in this setting in order to ensure that inferences rest squarely on the design.
If assumptions about sampling procedures or data generating processes are necessary to ensure the reliability of estimation, then inferences are not based solely on properties of the experiment, but rather on factors outside the control of the experimenter~\citep{aronow2021nonparametric}.

Fourth, units which do not show up in the sample should not be included in the balancing.
An offline algorithm used in an online environment would fail this condition, because to use it would require to balance the entire population of units \emph{expected} to show up to an experiment.
Units within that population who fail to show would nevertheless be assigned a treatment.
Depending on the distribution of units who actually show up in the sample, there are no longer any guarantees that balance will obtain.

Our approach satisfies all four of these desiderata.
We propose an online method for covariate balancing requiring linear time and space which provably provides variance reduction. 
Building upon recent work on discrepancy minimization---the self-balancing walk~\citep{alweiss2020discrepancy} and kernel thinning~\citep{dwivedi2021kernel}---we provide an algorithm whose $L_2$-imbalance matches the best known \emph{online} discrepancy algorithm. 
Where $\delta$ is a failure probability, performing this optimization online results in a $\log(n/\delta)$ cost in convergence of the average treatment effect over the offline algorithm of discrepancy minimization by~\citet{harshaw2020balancing}.
A conjectured lower bound for online discrepancy minimization would further reduce this cost by a square root.
We also provide an extension to multiple treatments and non-uniform treatment probabilities.

The rest of the paper is organized as follows.
Section~\ref{sec:related} with a discussion of related work to position our contribution in the literature.
Section~\ref{sec:problem} defines notation and formally introduces the problem.
Section~\ref{sec:method} provides our proposed algorithms and methods.
Section~\ref{sec:simulations} provides a detailed simulation study of the behavior of the proposed algorithms.

\section{Related Work}
\label{sec:related}
There are two common approaches for achieving improved covariate balance in experiments.
The first, and most common especially within industrial settings, approach is to perform a post-hoc regression adjustment which includes pre-treatment covariates~\citep{deng2013improving, lin2013agnostic}.
The second approach is to consider covariate balance during the design phase of the experiment, i.e., explicitly optimizing treatment assignment in order to minimize imbalance between treatment groups~\citep{greevy2004optimal, higgins2016improving, kallus2018optimal}. 

Post-hoc stratification may be seen as asymptotically equivalent in terms of variance reduction to its analogous pre-stratified design as shown by \citet{miratrix2013adjusting}.
\posscite{miratrix2013adjusting} analysis is limited by two factors: it assumes a fixed number of stratification cells (that do not grow with sample size) and it is conditioned on the post-stratification estimator being defined (e.g. treatment and control units within each stratification cell).
These limitations may weaken it's asymptotic equivalence argument.
A key limitation of post-hoc adjustment approaches is that the desire for simplicity and scalability implies that practitioners typically adjust for only a linear function of the pre-treatment covariates.
Indeed, in the common ``CUPED'' approach, adjustment is performed solely on a linear function of a single pre-treatment outcome measurement~\citep{deng2013improving}.
Second, many common approaches for constructing stratification cells (i.e. clustering algorithms) may be computationally infeasible in practice for industrial applications when the number of simultaneous experiments and the number of outcome variables of interest are large.
Third, without sample splitting (or when naively applied) advanced machine-learning based methods for adjustment may slip in assumptions of correct-specification of the outcome model, or have confidence intervals with poor coverage properties.
While cross-fitting may ameliorate some of these problems in larger samples, sample splitting may prove too high a cost when sample sizes are low.

A key shortcoming of design-based covariate balance is the lack of computationally efficient algorithms which provide theoretical guarantees over worst case behavior.
Blocking~\citep{fisher1935design} partitions variables into non-overlapping sets and performs complete randomization within each partition~("blocks").
\citet{higgins2016improving} introduced a computationally approximation of blocking which runs in $\mathcal{O}(n\log(n))$ time. 
\citet{kallus2018optimal, bertsimas2015power} propose an optimization based approach, \citet{kallus2018optimal} additionally considers a partially random approach using semi-definite programming.
\citet{zhou2018sequential} provide a method combining \emph{batch-based} sequential experimentation with rerandomization to achieve balance, but which is not computationally feasible in moderate to large sample sizes.
Perhaps the closest to the current work is \citet{harshaw2020balancing} which propose a balancing design using the Gram-Schmidt walk, an offline method for (linear) discrepancy minimization.
Current state of the art for balancing treatment assignment requires polynomial running time and generally requires knowing all of the covariate vectors prior to determining assignment~\citep{higgins2016improving, harshaw2020balancing,arbour2021efficient}.
As we discuss in section \ref{sec:problem}, this is a non-starter for online treatment assignment.
In this setting, subjects must be allocated as they arrive; it does no good to know how you \emph{should have} assigned a user at the end of the experiment; you need to know when that subject arrives.
It's crucial that when a subject in an experiment arrives they be swiftly allocated to a unit.
Especially in an online environment, high latency will lead to attrition, which may counteract any potential gains from greater efficiency.
Moreover, the users who attrit may be the very subjects of interest~\citep{munger2021accessibility}.
By inducing differential attrition based on patience, the sample in the experiment may differ greatly from the population of interest on unobserved characteristics that make it difficult to extrapolate to a population-level effect~\citep{egami2020elements}.

There is also a variety of methods aimed at sequential, online assignment in experiments.
The seminal work in this literature is \citet{efron1971forcing} which introduced an online variant of complete randomization which aims to ensure that a pre-specified marginal treatment probability is met without introducing too much accidental bias.
\citet{smith1984sequential} provides a generalization of the \citet{efron1971forcing} approach which extends gracefully to multiple treatments.
There are a variety of online balanced coin designs which seek to reduce covariate imbalance~\citep[e.g.][]{baldi2011covariate, moore2017blocking}.
\citet{moore2017blocking} is based around Mahalanobis distance.
As such, it has polynomial time-complexity \emph{at each arrival time}.
In addition to inefficiency, the theoretical worst-case behavior of this algorithm has not been resolved, even in the stochastic setting.
Theoretical guarantees of this sort are paramount in the design setting, as practitioners need to know the credibility of their inferences (and how they may differ from simple Bernoulli randomization).

\section{Background and Problem Setting}
\label{sec:problem}

We first fix notation. 
Random variables will be denoted in upper case, with sets in bold.
The problem setting, which we refer to as experimental treatment allocation, is as follows.
We assume that we observe $1, \dots, n$ i.i.d. observations of $\XX \in \mathbb{R}^{n \times d}$: the covariates\footnote{We assume linear feature maps throughout. We note that nonlinearities can be handled with the same guarantees following \citet{dwivedi2021kernel} at the cost of additional computational complexity.}.
The experimenter is asked to assign a treatment assignment, $A \in \{1, -1\}$ (we will later loosen this to multiple discrete treatment values).
We will refer to the assignments of $A$ as treatment and control, respectively.
Each unit is imbued with \emph{potential outcomes} for each treatment, the value of the outcome if that unit had been assigned to the given group: $y(1)$ for treatment and $y(-1)$ for control.
After assignment we observe only the potential outcome corresponding to the realized treatment assignment, $Y$.
We assume that the outcomes are not available until the conclusion of the experiment.
At the end of the experiment we are interested in measuring the sample average treatment effect~(SATE) between any two treatments, $k$ and $k'$ with the difference in means estimator:
\begin{align}
\label{eq:SATE}
    \hat{\tau}_{kk'} = \frac{1}{n}\sum_i^n \frac{A_i}{p(A_i)} Y_i
\end{align}
where $p(A_i)$ denotes the probability of assigning treatment $A_i$ to instance $i$.
Note that this is simply the difference-in-means rather than the more general Horwitz-Thompson estimator~\citep{horwitz1952generalization}, as the treatment probability is marginal rather than conditional.
More sophisticated estimators are usable in this setting~\citep[e.g.][]{tsiatis2008,aronow2013}, but we will focus on the simplest as we optimize design as is commonplace for studying design~\citep{kallus2017balanced, harshaw2020balancing}.

If propensity scores are constant, the estimator of the SATE given by equation~\ref{eq:SATE} will be unbiased and consistent for its oracle counterpart,
\begin{align}
    \tau_{kk'} = \frac{1}{n}\sum_i^n y_i(k) - y_i(k'),
\end{align}
the difference of potential outcomes of the $k$ and $k'$ treatments.
This SATE is our estimand of interest, as estimated by equation~\ref{eq:SATE}.

We will maintain the following assumptions:
\begin{assumption}[Consistency] $Y_i = y_i(k)$ if $A_i = k \quad \forall i, k$.
\end{assumption}

The problem of experimental allocation is to observe covariate vectors and assign $A$ to units so as to achieve desirable properties of the SATE (for instance, to minimize variance).
In the most general setting where no assumptions are placed the  relationship between the covariates and outcome complete randomization---randomly drawing assignments without respect to background covariates---is known to be minimax optimal~\citep{kallus2018optimal}.

\subsection{Robustness in sequential design}

Experiments are prized for their ability to provide unbiased estimates of causal effects with relatively mild assumptions.
These assumptions, in fact, typically flow from the \emph{design} of the experiment rather than more difficult assumptions about the data used in the course of analysis~\citep{sekhon2009opiates,aronow2021nonparametric}.

In the study of vector balancing, there are three main sampling schemes of interest, listed in order of how adversarial they are: \begin{enumerate*}
\item Stochastic arrivals.
\item Oblivious adversarial arrivals.
\item Fully adversarial arrivals.
\end{enumerate*}

In the stochastic arrivals setting, units are sampled i.i.d. from some fixed (possibly infinite) population. 
As such, a given covariate vector is just as likely to arrive early in the sequence as late.
In the oblivious adversarial setting, the adversary knows the process which will be used to assign units to groups, but cannot condition specifically on those assignments in making its decisions about the order of arrival of units.
This implies the following conditional independence:
\begin{assumption}[Oblivious Adversary]
$$
    \ww_{i-1} \indep \xx_i | \mathcal{H}_{i-1}
$$
where $\mathcal{H}_i$ is the history of $\xx$ up to and including $\xx_i$.
\end{assumption}
This stands in contrast to the fully adversarial case, in which the adversary is able to condition its decision in each period on the full set of past assignments in addition to the history of the covariate vector.
We will focus on the oblivious adversary in this paper.

\section{Online Assignment of Treatments}
\label{sec:method}
\textbf{Weighted Online Discrepancy Minimization:} We consider a variant of the online Koml\'os problem~\citep{spencer1977balancing}, where vectors $\xx_1, \ldots \xx_n$ arrive one by one and must be immediately assigned a weighted sign of either $-2q$ or $2(1-q)$, for $0 < q < 1$, such that the weighed discrepancy $\left\| \sum_{i=1}^n \eta_i \xx_i \right\|_{\infty}$, where $\eta_i$ is the weighted sign given to $\xx_i$, is minimized.
Notice that when $q = 1/2$, the signs become $\pm 1$.

Algorithm \ref{alg:weighted}, takes $i=1,\dots,n$ unit vectors in sequentially and assigns them to a treatment and control, 
represented by the value of $\eta_i$.
The procedure, an extension of recent work in online discrepancy minimization~\citep{alweiss2020discrepancy, dwivedi2021kernel}, assigns treatment with probability proportional to the inner product between a running sum of the signed prior observations.
The algorithm and analysis differs from prior work for discrepancy in two aspects which are necessary for use in experimentation. One is to give a ridge regression guarantee, we need to characterize the random vectors output by the algorithm in terms of the projection matrix $\PP.$ \citet{harshaw2020balancing} do that for an offline discrepancy minimization algorithm, and here we do it for an online version. The other way we differ is that our algorithm is a slight generalization of the one in \citet{dwivedi2021kernel} in that we have a parameter $q$. This allows for the case of imbalanced treatment assignments. A straightforward adoption of the analysis in \citet{dwivedi2021kernel} to this case results in a worse dependence on $1/q.$ Therefore, we derive a sub-exponential concentration bound and get a better dependence on $1/q.$

\begin{algorithm}
\label{alg:main}
\caption{takes each input vector ${\xx}_i$ and assigns it $\{-2q, 2(1-q)\}$ signs online to maintain low weighted discrepancy with probability $1-\delta$.}
\begin{algorithmic}
\STATE {\bfseries Input: } {${\xx}, q$}

$c \gets \min(1/q, 9.3) \log(2n/\delta)$ \\
\FOR {$i$ from 1 to $n$} 
\IF {$|\ww_{i-1}^{\top} {\xx}_i| > c$} 
\STATE ${\ww_{i}} \gets {\ww_{i-1}} - 
2 q \frac{\ww_{i-1}^{\top} {\xx}_i}{c} \xx_i $
\ELSE
\STATE $\eta_i \gets \begin{cases}
     2(1-q), & \text{ w.p. }  q(1 - \ww_{i-1}^{\top} {\xx}_i / c)\\
     -2q,& \text{ w.p. } 
     1- q(1 - \ww_{i-1}^{\top} {\xx}_i / c)
     \end{cases}
$
\STATE ${\ww_{i}} \gets {\ww_{i-1}} + \eta_i {\xx_i}$
\ENDIF

\ENDFOR

\STATE {\bfseries Output: } {$\eeta, \ww$}
\end{algorithmic}

\label{alg:weighted}
\end{algorithm}

{\bf Notation: }Let $\PP_i, \, i\in [n]$ be orthogonal projection matrices onto the span of $\{\xx_1,..,\xx_i\},$ that is, 
\[\PP_i = \XX_i^{\top} \br{\XX_i \XX_i^{\top}}^{-1} \XX_i ,\]
where $\XX_i$ is the $i \times d$ submatrix of $\XX$ corresponding to covariates $\{\xx_1,..,\xx_i\}.$ Let $A = 0.5803, B = 0.4310$ and $ \alpha = 2/B $.

\begin{definition}
[Sub-Gaussian]
A mean zero random variable $X$ is sub-Gaussian with parameter $\sigma$ if for all $\lambda \in \mathbb{R}$,
$
\E [ \exp\br{ \lambda X } ] \leq \exp \br{ \frac{\lambda^2 \sigma^2}{2}}.
$

A mean zero random vector $\ww$ is $(\sigma, \PP)$ sub-Gaussian if for all unit vectors $\uu$ and $\lambda \in \mathbb{R}$,
$
\E [ \exp\br{ \lambda \ww^{\top} \uu } ] \leq \exp \br{ \frac{\lambda^2 \sigma^2 \uu^{\top} \PP \uu }{2}}.
$
In particular, $\ww^{\top} \uu$ is $\sigma^2 \uu^{\top} \PP \uu$ sub-Gaussian.
\end{definition}

\begin{definition}
[Sub-exponential] 
A mean zero random variable $X$ is $(\nu, \alpha)$ sub-exponential if for all $|\lambda| \leq \frac{1}{\alpha}$,
$
\E [ \exp\br{ \lambda X } ] \leq \exp \br{ \frac{\lambda^2 \nu^2}{2}}.
$

A mean zero random vector $\ww$ is $(\nu, \alpha, \PP)$ sub-exponential if for all unit vectors $\uu$ and $|\lambda| \leq \frac{1}{\alpha \sqrt{ \uu^{\top} \PP \uu}}$,
$
\E [ \exp\br{ \lambda \ww^{\top} \uu } ] \leq \exp \br{ \frac{\lambda^2 \nu^2 \uu^{\top} \PP \uu }{2}}.
$
In particular, $\ww^{\top} \uu$ is $(\nu \sqrt{\uu^{\top} \PP \uu}, \alpha \sqrt{\uu^{\top} \PP \uu})$ sub-exponential.
\end{definition}

\begin{restatable}[Main]{theorem}{main}
\label{thm:main}
Let $\ww_1, ... \ww_n$ be as in Algorithm~\ref{alg:main}. Then \begin{enumerate}
    \item $\ww_i$ is mean zero $\br{\sqrt{c/2q}, P_i}$ sub-Gaussian.
    \item $\ww_i$ is mean zero $\br{\sqrt{8Ac}, \alpha, P_i}$ sub-exponential.
    \item With probability $1-\delta,$ for all $i,$  $|\ww_i^T \xx_i| \leq c .$ 
\end{enumerate}
\end{restatable}
Note that $\eta_i$ is defined only when $|\ww_{i-1}^{\top} \xx_i| \leq c$. Therefore, $\eeta$ is defined with probability at least $1-\delta.$
\begin{remark}
If $\ww_i$ is a is a mean zero $\br{\sigma, \PP_i}$ sub-Gaussian random vector, then $\Cov(\ww_i) \leq \sigma^2 \PP_i.$ Similarly, we have that if $\ww_i$ is a is a mean zero $\br{\nu, \alpha, P_i}$ sub-exponential random vector, then $\Cov(\ww_i) \leq \frac{3}{2} \nu^2 \PP_i$.
A proof is given in Lemma~\ref{lem:vector-var}.
\end{remark}
We will use Threorem~\ref{thm:main} to derive results on the average treatment effect using the framework developed by \cite{harshaw2020balancing}.

Let $z_i = \eta_i + 2q - 1 \in \{-1, 1\}.$ We will show that $\E[\eta_i] = 0,$ and so we have $\E[z_i] = 2q-1$ and $\zz - \E[\zz] = \eeta.$ Let $ \mmu = \frac{\boldsymbol{Y}(1)}{4 q} + 
\frac{\boldsymbol{Y}(0)}{4(1-q)}.$
\citet{harshaw2020balancing} give a linear algberaic expression for the error of HT-estimators in terms of $\mmu$. In particular, they show 
\begin{lemma}[Lemma A2 and Corollary A1 in \cite{harshaw2020balancing}]
\label{lem:ate-err}
\[
\hat{\tau} - \tau = \frac{2}{n} \br{\zz - \E[\zz]}^{\top} \mmu = \frac{2}{n} \eeta^{\top}\mmu
\]
and hence,
\[ \Var(\hat{\tau}) = \E \brs{ (\hat{\tau} - \tau)^2 } = \frac{4}{n^2} \mmu^{\top} \Cov(\zz) \mmu.
\]
\end{lemma}

Theorem~\ref{thm:main} immediately implies that assignments generated by Algorithm~\ref{alg:main} are well balanced. But optimizing just for balance can lead to accidental bias~\citep{efron1971forcing}.  We have from Lemma~\ref{lem:ate-err} that $\text{Var}(\hat{\tau}) = \frac{4}{n^2} \lambda_{\max}\br{\Cov(\zz)} \|\mmu\|^2 $ in the worst case when $\mmu$ is along the top eigenvector of $\Cov(\zz)$. %
Therefore, to control accidental bias we need to make sure $\lambda_{\max}\br{\Cov(\zz)}$ is not high. 

We achieve this by augmenting the original covariates $\xx_i$ by $\sqrt{\phi} \ee_i$ to get 
$\left[ \begin{array}{c}
\sqrt{\phi} \ee_i \\
\sqrt{1-\phi} \xx_i
\end{array}
\right]
,$
where  $\phi \in [0, 1]$ is a parameter which controls the extent of the covariate balance, and $\ee_i$ is a basis vector in dimension $n$.  
By a simple calculation, we can see that running Algorithm~\ref{alg:main} on $\xx_i$ augmented with $\sqrt{\phi} \ee_i$ is equivalent to running it with $\xx_i$ and replacing $\ww_{i-1}^{\top} \xx_i$ by $\sqrt{1- \phi} \ww_{i-1}^{\top} \xx_i$ everywhere. Therefore, we don't have to explicitly augment the covariates in the algorithm. 

We note that with augmented covariates, $\ww_n = \left[\begin{array}{c}
\sqrt{\phi} \eeta \\
\sqrt{1-\phi} \boldsymbol{X}^{\top} \eeta
\end{array}\right]$ is a sub-Gaussian or a sub-exponential random vector as in Theorem~\ref{thm:main}.

{\small Let $\phi \in (0,1)$ be fixed, and $\QQ = \br{\phi  \II + (1-\phi)\XX \XX^{\top}}^{-1}.$}
 
\begin{proposition}
\label{prop:z-dist}
When Algorithm~\ref{alg:main} is run with augmented covariates as described above, then 
$\eeta = \zz - \E[\zz]$ is a mean zero $(\sqrt{c/2q}, \QQ)$ sub-Gaussian random vector and also, $\eeta$ is a mean zero $(\sqrt{8Ac},\alpha, \QQ)$ sub-exponential random vector.
\end{proposition}

Proposition~\ref{prop:z-dist} shows that $\E[\eeta] = 0$ for all $i.$ 
\begin{proposition}[Unbiasedness]
When Algorithm~\ref{alg:main} is run with augmented covariates, we have with probability at least $1-\delta$, $\E \brs{\sum_i \xx_i \eta_i} = 0.$
\end{proposition}

For rest of the section, we let $\sigma^2 = c/2q$ if $c = \log(2n/\delta)/q$ and $\sigma^2 = 12Ac$ if $c = 9.3 \log(2n/\delta).$ 
\begin{proposition}[Eigenvalues of Treatment Covariance]
\label{prop:spectral-bound}
With probability at least $1-\delta$, Algorithm~\ref{alg:main} produces $\eeta$ satisfying
$
\Cov(\zz) = \Cov(\eeta) \preceq \sigma^2 \QQ.
$

\end{proposition}

\subsection{Balance}
\label{sec:balance}

\begin{proposition}
\label{prop:balance}
Let $\ww = \sum_i \eta_i \xx_i$. With probability at least $1-\delta$,
{
\small
\[
\left \| \ww \right \|_{2} \leq \sqrt{d} \left \| \ww \right \|_{\infty} \leq 
\min \left( \frac{1}{q}, 9.3 \right) \sqrt{\frac{ {d \log (4d/\delta)\log(4n/\delta)} }{2 (1-\phi) \phi}}.
\]
}
\end{proposition}

\subsection{Error bounds}

\begin{proposition}[Concentration of ATE]
\label{prop:ate-concentration}
Algorithm~\ref{alg:main} when run with augmented covariates, generates a random assignment $\zz$ such that
\[
|\hat{\tau} - \tau| = \frac{2}{n} | \eeta^{\top} \mmu| \leq \frac{2c}{n} 
\sqrt{\mmu^{\top} \QQ \mmu}.
\]
with probability $1-\delta.$ 
\end{proposition}
Lemma A10 in~\cite{harshaw2020balancing} shows that $\mmu^{\top} \QQ \mmu = \min _{\boldsymbol{\beta} \in \mathbb{R}^{d}}\left[\frac{1}{\phi }\|\boldsymbol{\mu}-\boldsymbol{X} \boldsymbol{\beta}\|^{2}+\frac{\|\boldsymbol{\beta}\|^{2}}{(1-\phi) } \right].$

\begin{proposition}
\label{prop:ate-worstcase}
The worst-case mean squared error of the online balancing walk design is upper bounded by %
\begin{align*}
&\mathrm{E}\left[(\widehat{\tau}-\tau)^{2}\right] \leq \frac{4 \sigma^2}{\phi n^2} \sum_{i=1}^{n} \mu_{i}^{2}\\ %
\end{align*}
where $\phi \in(0,1]$ with probability $1 - \delta$.
\end{proposition}
\begin{proof}[Proof of Proposition~\ref{prop:ate-worstcase}]

This follows from Lemma~\ref{lem:ate-err} and Proposition~\ref{prop:spectral-bound}. We note that $\QQ \preceq \frac{\sigma^2}{\phi} \II.$
\end{proof}
\begin{proposition}[Ridge Connection]
\label{prop:ate-ridge}
The worst-case mean squared error of the online balancing walk design is upper bounded by an implicit ridge regression estimator with regularization proportional to $\phi$. That is, %
\begin{align*}
    &\E \left[(\widehat{\tau}-\tau)^{2}\right] \leq \frac{4 \sigma^2 L}{n} \\
    &\text { where } \quad L=\min _{\boldsymbol{\beta} \in \mathbb{R}^{d}}\left[\frac{1}{\phi n}\|\boldsymbol{\mu}-\boldsymbol{X} \boldsymbol{\beta}\|^{2}+\frac{\|\boldsymbol{\beta}\|^{2}}{(1-\phi) n}\right]
\end{align*}
with probability $1 - \delta$.
\end{proposition}

\subsection{Algorithm with Restart}
\label{sec:restart}
We saw earlier that $\eeta, \zz$ are defined only with probability $1-\delta.$ This is because with our choice of $c$, only with probability $1-\delta$ we have for all $i, |\ww_i^{\top} \xx_i| \leq c.$ This means our treatment assignment fails with probability $\delta.$ There is a simple way to make sure that the algorithm never fails and have same error bound guarantees with a slightly worse constant. 

This is achieved by slightly modifying Algorithm~\ref{alg:main} so that whenever we have $|\ww_{i-1}^{\top} \xx_i| > c$ for a particular $i,$ we start a new instance of the algorithm for covariates $\xx_{i+1}, ...\xx_n.$ This is equivalent to setting $\ww_{i-1} = 0$ and continuing with the algorithm. 

Since for any treat assignment procedure $\E \br{\hat{\tau}-\tau}^2$ just depends on $\Cov(\zz)$, and Algorithm~\ref{alg:main} fails with probability $\geq \delta,$ we can show that 
\begin{align*}
  \Cov(\zz) &\preceq (1-\delta)\QQ  + \delta \QQ + \delta^2 \QQ + ...\\
  &\preceq 2 \QQ \text{ when } \delta \leq 1/2.
\end{align*}

Therefore, for the modified algorithm, we will have error guarantees as in Propositions~\ref{prop:ate-worstcase} and \ref{prop:ate-ridge} (but worse by at most a factor of $2$) and with probability one.
\subsection{Multiple Treatments}

In this section, we consider an online multi-treatment setting, where each vector is assigned to a group in  $M = \{m_1, m_2, \ldots, m_k\}$ immediately on arrival. For each $1 \leq i\leq k$, group $m_i$ is associated with a weight $\alpha_i$. 
The goal is to minimize the multi-treatment discrepancy:
\begin{align*}
&\max_{m_1, m_2 \in M} 2 \left\| \frac{\sss(m_1)/{\alpha_1} - \sss(m_2)/{\alpha_2}}{ {1}/{\alpha_1} + {1}/{\alpha_2}}  \right\|_{\infty} 
\end{align*}
where $\sss(m)$ is the sum of all vectors assigned to treatment $m$.
Notice that by setting $\alpha_1 = \frac{1}{1-q}$ and $\alpha_2 = \frac{1}{q}$, 
we can recover the definition given for the weighted discrepancy between two treatments $m_1$ and $m_2$.

Our algorithm can leverage any oracle (we call it \texttt{BinaryBalance} in Algorithm~\ref{alg:multi}) that minimizes the weighted discrepancy for two treatments. Our results are obtained by using Algorithm~\ref{alg:weighted}.

We first build a binary tree where each leaf of the tree corresponds to one of the $k$ treatments in $M$. Let $h$ be the smallest integer such that $2^{h} \geq k$. We start with a complete binary tree of height $h$, and then remove $2^h - k$ leaves from the tree such that no two siblings are removed. Note that this is possible by the definition of $h$. We further contract each internal node with only one child to its child. This process does not change the number of leaves in the tree.  Let $T$ be the obtained tree. By construction, all internal nodes of $T$ have 2 children and $T$ has exactly $k$ leaves. We then associate each leaf of $T$ with a treatment in $M$. For each vector assigned to treatment $m_i$, we also say that it is assigned to the leaf corresponding to $m_i$, 
$\forall 1 \leq i \leq k$.
For each node $v \in T$, denote by $\sss(v)$  the sum of all vectors assigned to leaves under $v$. 
In addition, let $\alpha(v)$ be the sum of all weights assigned to leaves under $v$.
For each internal node $v$ of $T$, the weighted discrepancy vector at $v$ is defined as: 
\begin{align*}
    \ww(v) &=  \frac{\alpha(v_r)}{\alpha(v_l) + \alpha(v_r)} \sss(v_l) - \frac{\alpha(v_l)}{\alpha(v_l) + \alpha(v_r)} \sss(v_r) \\
    &= \frac{{\sss(v_l)}/{\alpha(v_l)}  - {\sss(v_r)}/{\alpha(v_r)}}{{ 1}/{\alpha(v_l)} + { 1}/{\alpha(v_r)}},
\end{align*}
where $v_l$ and $v_r$ are the left and right child of $v$ respectively.

For each internal node $v$ in $T$, we maintain an independent run of a two-treatment algorithm that minimizes $\|\ww(v)\|_{\infty}$. At a high level, we minimize the weighted discrepancies at all internal nodes simultaneously. When a new vector $\xx$ arrives, we first feed it to the algorithm at root $r$. If the result is $+$, we continue with the left sub-tree of $r$. Otherwise, we go to the right sub-tree. We continue in that manner until we reach a leaf $l$. $\xx$ will then be assigned to $l$ (and the treatment associated with $l$).
\begin{restatable}{theorem}{multi}
\label{thm:multi}
Let \texttt{BinaryBalance} be Algorithm~\ref{alg:weighted}. Then Algorithm~\ref{alg:multi} obtains 
$O\left(\log k \sqrt{\frac{ {(1-\phi)d \log (dk/\delta)\log(nk/\delta)} }{\phi}} \right)$ 
multi-treatment discrepancy with probability $1-\delta$. 
\end{restatable}

\begin{algorithm}[h]
\caption{\texttt{KGroupBalance} takes each input vector ${x}_i$ and assigns it to one of the groups online to maintain low discrepancy with probability $1-\delta$.}
\begin{algorithmic}
\STATE {\bfseries Input:} ${\xx}, k, \alpha$.
\STATE $h$ $\gets$ smallest integer such that $2^h \geq k$. 

\STATE $T \gets$ complete binary tree with height $h$. Remove $2^h - k$ leaves from $T$ such that no two siblings are removed. Associate each treatment to a leaf of $T$. Contract each internal node in $T$ with one child to its child.

\FOR{node $v$ in $T$}
    \STATE $\alpha(v) \gets$ sum of all weights of groups associating with leaves under $v$.
\ENDFOR

\FOR{internal node $v$ of $T$}
    \STATE Instantiate \texttt{BinaryBalance}($v$) $\gets$ oracle for weighted discrepancy problem at $v$ with weighted signs $\frac{\alpha(v_r)}{\alpha(v_l) + \alpha(v_r)}$
    and $-\frac{\alpha(v_l)}{\alpha(v_l) + \alpha(v_r)}$.
\ENDFOR

\FOR{i from 1 to n}

\STATE $v \gets $ root of $T$. 

\FOR{$v$ is an internal node of $T$}

\STATE Feed $\xx_i$ to \texttt{BinaryBalance}($v$)

\STATE $v \gets $ one of the children of $v$ according to the assignment of \texttt{BinaryBalance}($v$) on input $\xx_i$
\ENDFOR

\STATE Assign $x_i$ to the group corresponding to $v$.

\ENDFOR
\end{algorithmic}

\label{alg:multi}
\end{algorithm}

\begin{figure}
    \centering
    \includegraphics[width=0.375\textwidth]{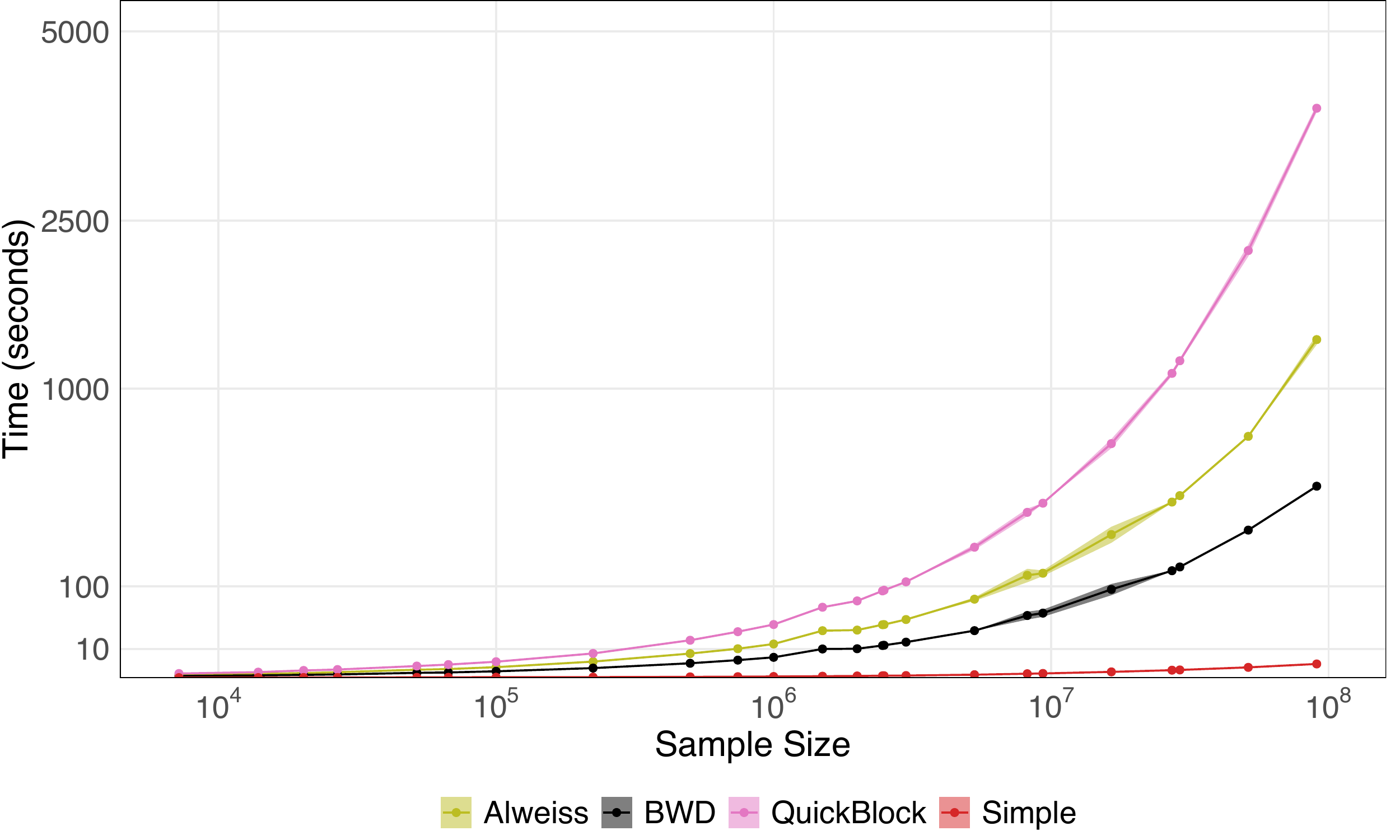}
    \caption{Time to design. All timings performed on a \texttt{ml.r5.2xlarge} instance of Amazon SageMaker. The y-axis is scaled by the square root for easier visualization.}
    \label{fig:time}
\end{figure}

\section{Experiments}
\label{sec:simulations}

In this section, we provide simulation evidence on the efficacy of our proposed methods.
In particular, we use a wide variety of data generating processes, many of which do not assume that units arrive i.i.d. as is often standard in simulation settings for this problem.
Subjects are unlikely to arrive truly i.i.d. in the real world.
Earlier arrivals will typically be more active than late-arriving units, for example.

All data generating processes used in simulations are shown in Table~\ref{tab:dgps}.
If not otherwise specified, the sample size is $1000$ subjects, the number of groups is two and the marginal probability of treatment is $\frac{1}{2}$.

Methods compared in the simulations are simple randomization (Bernoulli coin flips), complete randomization (fixed-margins randomization), the biased coin designs of \citet{efron1971forcing} and of \citet{smith1984sequential}, QuickBlock of \citet{higgins2016improving}, and \citet{alweiss2020discrepancy}.
These are compared to our proposed methods which are generalized versions of the discrepancy minimization procedure of \citet{dwivedi2021kernel}.
We provide three versions of our proposed algorithms, called BWD for "Balancing Walk Design". 
The most basic version (BWDRandom) reverts to simple random assignment for all remaining periods once $|\ww_{i-1}^{\top} \xx_i| > c$.
Our preferred approach restarts the algorithm in this case as described in section~\ref{sec:restart}.
We examine this with two levels of robustness, $\phi$: 0 (purely balancing) and 0.5 (a uniform mix between randomization and balancing): BWD(0) and BWD(0.5) respectively. For further details, see section~\ref{sec:robust}.

In these comparisons, BWD need not out-perform all methods in all data-generating processes.
QuickBlock, for instance, is a fully off-line method, so comparable performance by an online method is noteworthy.
In general, Alweiss and BWD will be most effective when the true relationship between covariates and outcome is linear, since they seek linear balance.
All plots incorporate 95\% confidence intervals.
\begin{figure}
    \centering
    \includegraphics[width=0.4\textwidth]{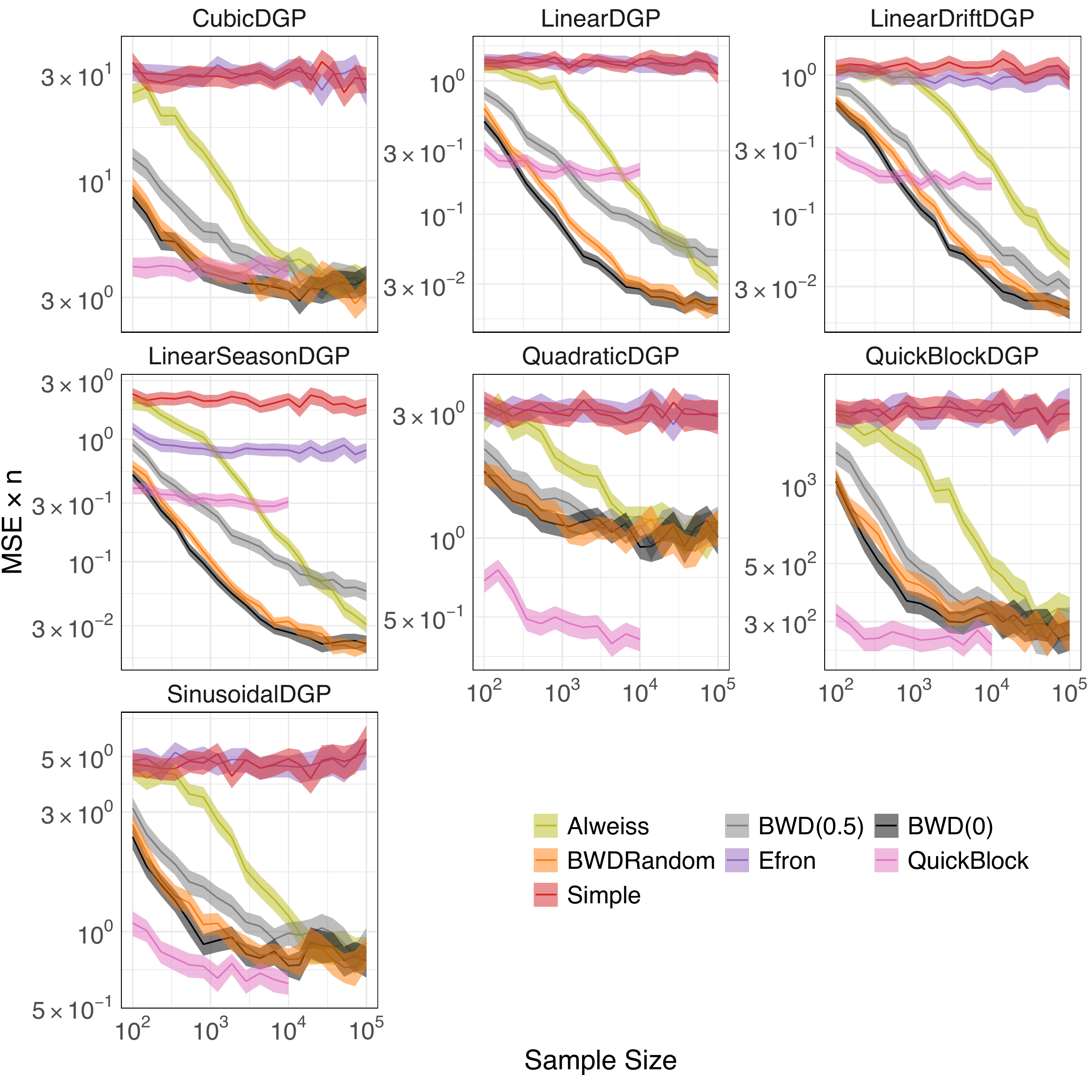}
    \caption{MSE. BWD provides effective variance reduction across a wide array of simulation environments.}
    \label{fig:mse}
\end{figure}
\subsection{Binary treatment}
\paragraph{Timing.} Our proposed method (BWD) is highly efficient, scaling substantially better than other balancing methods.
This analysis of runtime directly compares online methods to a widely used offline balancing method (QuickBlock).
It's important to note that the QuickBlock algorithm cannot be used in the online setting, even if its runtime did not make that prohibitive.
While QuickBlock is $\mathcal{O}(n \log n)$, the proposed online balancing methods are all linear-time.
Given QuickBlock's runtime, the following simulations only include it for comparisons up to sample sizes of $10^4$.

\paragraph{MSE.}
Next, we demonstrate in Figure~\ref{fig:mse} how imbalance minimization translates to improved estimation of causal effects, measured by the mean squared-error of our estimate of the average treatment effect.
We normalize this graph based on $n$, the rate of convergence of the difference-in-means estimator under simple randomization.
The results depend strongly on the true nature of the data-generating process.
In short, on non-linear data generating processes, offline blocking performs better than anything else, but in many settings BWD converges to similar error rates as QuickBlock.
On linear or near-linear data-generating processes, our proposed algorithms perform very strongly, outperforming QuickBlock even in small sample-sizes.
When there is a break from the purely i.i.d. stochastic setting (such as \texttt{LinearDriftDGP} and \texttt{LinearSeasonDGP}), BWD behaves well, as expected.

\subsection{Multiple Treatments}

\paragraph{MISE.}
BWD gracefully extends to the multiple-treatment setting, which we demonstrate in Figure~\ref{fig:mse-multi}.
This chart measures the mean integrated squared-error of our estimates of the ATEs (relative to a single control group).
Figure~\ref{fig:mse-multi} shows the results.
BWD consistently outperforms existing online assignment methods by substantial margins.

\begin{figure}
    \centering
    \includegraphics[width=0.4\textwidth]{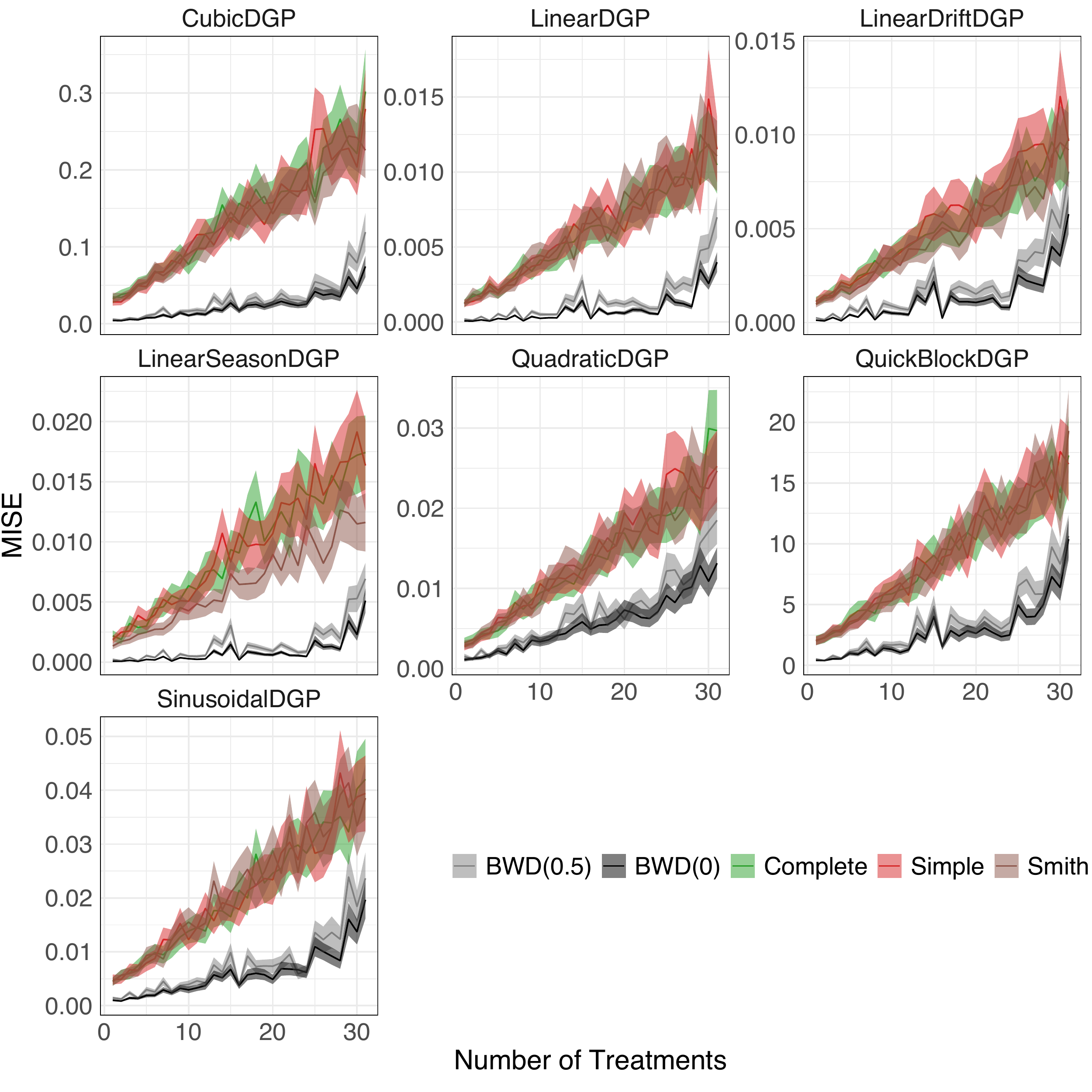}
    \caption{MISE. BWD effective reduces variance no matter the number of treatments.}
    \label{fig:mse-multi}
\end{figure}

An array of additional simulation results may be found in Appendix~\ref{app:sims}.

\section{Conclusion}

Experiments are a crucial part of how humans learn about the world and make decisions.
This paper is aimed at providing a way to more effectively run experiments in the online setting.
Practitioners must commonly operate their experiments in this environment, but due to the lack of suitable options for design fall back to simple randomization as the assignment mechanism.
In this paper, we have shown how the Balancing Walk Design can be an effective tool in this setting.
It is efficient, effective at reducing imbalance (and, therefore, the variance of resulting causal estimates), robust and it is fully suited to the particularities of online treatment assignment.
Simulations have shown it to work well across a range of diverse settings.
The Balancing Walk Design can improve the practice of large-scale online experimentation.

\begin{appendix}
\setcounter{figure}{0}
\setcounter{table}{0}
\renewcommand\thefigure{A\arabic{figure}} 
\renewcommand\thetable{A\arabic{table}} 

\section{Proofs}

\subsection{Proof of Theorem~\ref{thm:main}}

\begin{restatable}[Loewner Order]{proposition}{loewner}
\label{prop:loewner}
Let $\MM \succeq 0, C \geq 1, L \geq 0 .$ If $$\MM \preceq CL\boldsymbol{P} \coloneqq CL\boldsymbol{B}^\top\left(\boldsymbol{B}^\top\boldsymbol{B}\right)^{-1}\boldsymbol{B}$$ then for any vector $\vv \in \mathbb{R}^{n}$ with $\|\vv\|_{2} \leq 1$
$$
\MM^{\prime}=\left(\II-C^{-1} \vv \vv^{\top}\right) \MM\left(\II-C^{-1} \vv \vv^{\top}\right)+L \vv \vv^{\top}
$$
satisfies $$0 \preceq \MM^{\prime} \preceq CL\boldsymbol{P}' \coloneqq CL\boldsymbol{P} + CL\frac{\left(\II - \boldsymbol{P}\right)\vv \vv^\top\left(\II - \boldsymbol{P}\right)}{\vv^\top\left(\II - \boldsymbol{P}\right)\vv}.$$
\end{restatable}
\begin{proof}
By definition of $\MM'$ and the assumption $\MM \preceq cL\boldsymbol{P},$ we have 
\[
\MM^{\prime} \preceq CL \left(\II-C^{-1} \vv \vv^{\top}\right) \boldsymbol{P}\left(\II-C^{-1} \vv \vv^{\top}\right)+L \vv \vv^{\top}.
\] 
Therefore, it is sufficient to prove
\[
\left(\II-C^{-1} \vv \vv^{\top}\right) \boldsymbol{P}\left(\II-C^{-1} \vv \vv^{\top}\right)+ C^{-1} \vv \vv^{\top} \preceq
 \boldsymbol{P} + \frac{\left(\II - \boldsymbol{P}\right)\vv\vv^\top\left(\II - \boldsymbol{P}\right)}{\vv^\top\left(\II - \boldsymbol{P}\right)\vv}.
\]
Also note that, since $\|\vv\| \leq 1$ and $C \geq 1,$ we can absorb $C$ into $\vv$ (by taking $\vv \coloneqq \vv/\sqrt{C}$) and therefore without loss of generality assume that $C=1$ and $\|\vv\| \leq 1.$
Define:
\begin{align*}
    &\Px = \boldsymbol{B}^\top\left(\boldsymbol{B}^\top\boldsymbol{B}\right)^{-1}\boldsymbol{B}\quad
    &\Qx \coloneqq \left(\II - \boldsymbol{P}\right)\quad
    &\Pv = \vv\vv^\top\quad
    &\Qv = \left(\II - \vv\vv^\top\right)\\
    &\aaa \coloneqq \boldsymbol{P}_x \vv\quad
    &\bb \coloneqq \boldsymbol{Q}_x \vv\quad
    &\alpha \coloneqq \|\aaa\|^2\quad
    &\beta \coloneqq \|\bb\|^2
\end{align*}
We want to show
\begin{align*}
    \Qv\Px\Qv + \Pv \preccurlyeq
    \Px + \frac{\Qx\Pv\Qx}{\beta}
\end{align*}
Beginning by rewriting the LHS
\begin{align*}
    &\Qv\Px\Qv + \Pv = (\II - \Pv)\Px(\II - \boldsymbol{P}_v) + \Pv\\
    &=\Px - \Pv\Px - \Px\Pv + \Pv\Px\Pv + \Pv\\
    &=\Px - \vv\aaa^\top - \aaa\vv^\top + \vv\vv^\top\alpha + \Pv\\
    &=\Px - (\aaa + \bb)\aaa^\top - \aaa(\aaa + \bb)^\top + (\aaa + \bb)(\aaa + \bb)^\top\alpha + \Pv
\end{align*}
Expanding $\Pv$ as $\Pv = \left(\aaa + \bb\right)\left(\bb + \aaa\right)^\top = \aaa\aaa^\top + \bb\bb^\top + \aaa\bb^\top + \bb\aaa^\top$ gives
\begin{align*}
    &\Px - \left(\aaa + \bb\right)\aaa^\top - \aaa\left(\aaa + \bb\right)^\top + (1 + \alpha)\left(\aaa\aaa^\top + \bb\bb^\top + \aaa\bb^\top + \bb\aaa^\top\right)\\
    = &\Px + (\alpha - 1)\aaa\aaa^\top + \alpha\left(\aaa\bb^\top + \bb\aaa^\top\right) + (1 + \alpha)\bb\bb^\top
\end{align*}
Since $\| \vv\|^2 \leq 1$, we have $\alpha + \beta \leq 1.$ Now considering the difference of the LHS from the RHS after multiplying both sides by $\beta$ we arrive at
\begin{align*} \beta (\text{RHS} - \text{LHS}) =
    &\bb\bb^\top(\underbrace{1 - \beta(1 + \alpha)}_{1 - \beta - \beta\alpha \geq \alpha(1 - \beta) \geq \alpha^2}) + \beta(1 - \alpha)\aaa\aaa^\top - \alpha\beta(\aaa\bb^\top + \bb\aaa^\top)\\
    \succcurlyeq &\alpha^2\bb\bb^\top + \beta^2\aaa\aaa^\top - \alpha\beta(\aaa\bb^\top + \bb\aaa^\top)\\
    = &\left(\alpha \bb - \beta \aaa\right)\left(\alpha \bb - \beta \aaa\right)^\top \succcurlyeq 0.
\end{align*}
\end{proof}

\begin{lemma}
When $|\ww_{i-1}^{\top} {\xx}_i| < c$, we have
$$ \E \left[ \eta_i \right] = - 2 q \ww_{i-1}^{\top} {\xx}_i / c .$$
\end{lemma}
\begin{proof}
\begin{align*}
    \E \left[ \eta_i \right] &= 2(1-q) \cdot (q(1 - \ww_{i-1}^{\top} {\xx}_i / c)) +
    (-2q) \cdot (1- q(1 - \ww_{i-1}^{\top} {\xx}_i / c)) \\
    &= 0 + 2(1-q)(- q \ww_{i-1}^{\top} {\xx}_i / c) +
    (-2q)(q \ww_{i-1}^{\top} {\xx}_i / c) \\
    &= - 2 q \ww_{i-1}^{\top} {\xx}_i / c.
\end{align*}
\end{proof}

For all $i$ and $\uu \in \mathbb{R}^d$, we have
\begin{align*}
  \langle \ww_i, \uu \rangle &= \left \langle \ww_{i-1}, 
\uu - 2 q  \frac{\xx_i^{\top} \uu}{c} \xx_i \right \rangle + \epsilon_i \xx_i^{\top} \uu\\
  & = \langle \ww_i, \left( \II - \frac{2q}{c} \xx_i \xx_i^{\top} \right) \uu   \rangle + \epsilon_i \xx_i^{\top} \uu \\
  & = \langle \ww_i, \QQ_{\xx_i, c} \uu
  \rangle + \epsilon_i \xx_i^{\top} \uu,
\end{align*}
where 
$\epsilon_i =0$ if $|\ww_{i-1}^{\top}\xx_i| > c$ and $\epsilon_i = \eta_i - E [\eta_i]$ otherwise and $\QQ_{\xx_i, c} = \left( \II - \frac{2q}{c} \xx_i \xx_i^{\top} \right).$

Consider the case when $|\ww_{i-1}^{\top} \xx_i| \leq c$ and $\epsilon_i = \eta_i - E [\eta_i]$. Let $\tilde{q} = q(1 - \ww_{i-1}^{\top} {\xx}_i / c).$ Note that $0 \leq \tilde{q} \leq 2q.$ We have

$$\epsilon_i \gets \begin{cases}
     2(1 - \tilde{q}) & \text{ with probability }  \tilde{q},\\
     -2\tilde{q}& \text{ with probability } 
     1- \tilde{q}.
     \end{cases}
$$

\begin{definition}
[Sub-Gaussian]
A mean zero random variable $X$ is sub-Gaussian with parameter $\sigma$ if for all $\lambda \in \mathbb{R}$, 
\[
\E [ \exp\br{ \lambda X } ] \leq \exp \br{ \frac{\lambda^2 \sigma^2}{2}}.
\]

A mean zero random vector $\ww$ is $(\sigma, \PP)$ sub-Gaussian if for all unit vectors $\uu$ and $\lambda \in \mathbb{R}$,
\[
\E [ \exp\br{ \lambda \ww^{\top} \uu } ] \leq \exp \br{ \frac{\lambda^2 \sigma^2 \uu^{\top} \PP \uu }{2}}.
\]

In particular, $\ww^{\top} \uu$ is $\sigma'$ sub-Gaussian, where $\sigma'^2 = \sigma^2 \uu^{\top} \PP \uu$.
\end{definition}

\begin{definition}
[Sub-exponential] 
A mean zero random variable $X$ is $(\nu, \alpha)$ sub-exponential if for all $|\lambda| \leq \frac{1}{\alpha}$,
\[
\E [ \exp\br{ \lambda X } ] \leq \exp \br{ \frac{\lambda^2 \nu^2}{2}}.
\]

A mean zero random vector $\ww$ is $(\nu, \alpha, \PP)$ sub-exponential if for all unit vectors $\uu$ and $|\lambda| \leq \frac{1}{\alpha \sqrt{ \uu^{\top} \PP \uu}}$,
\[
\E [ \exp\br{ \lambda \ww^{\top} \uu } ] \leq \exp \br{ \frac{\lambda^2 \nu^2 \uu^{\top} \PP \uu }{2}}.
\]

In particular, $\ww^{\top} \uu$ is $(\nu', \alpha')$ sub-exponential, where $\nu'^2 = \nu^2 \uu^{\top} \PP \uu$ and $\alpha' = \alpha \sqrt{\uu^{\top} \PP \uu}$.
\end{definition}

The following concentration bounds for sub-Gaussian and sub-exponential vectors are obtained from standard bounds for scalar sub-Gaussian and sub-exponential random variables \cite{wainwright_2019} by scaling $\sigma$, $\nu$ and $\alpha$ by appropriate factors.  

\begin{lemma}[Sub-Gaussian Concentration]
\label{lem:subgaussian-bound}
If a random vector $\ww$ is $(\sigma, \PP)$ sub-Gaussian, then for all unit vectors $\uu$,
$$
P\br{|\ww^{\top}\uu| \geq t} \leq 
2 \exp \br{-\frac{t^2}{2\sigma^2 \uu^{\top} \PP \uu}}.
$$
\end{lemma}

\begin{lemma}[Sub-exponential Concentration]
\label{lem:subexp-bound}
If a random vector $\ww$ is $(\nu, \alpha, \PP)$ sub-exponential, then for all unit vectors $\uu$,
$$
P\br{|\ww^{\top}\uu| \geq t} \leq 
\begin{cases}
2\exp \br{-\frac{t^2}{2 \nu^2 \uu^{\top} \PP \uu}} & \text{ if } 0 \leq t \leq \frac{\nu^2 \sqrt{ \uu^{\top} \PP \uu}}{\alpha} \\
2\exp \br{- \frac{t}{2 \alpha \sqrt{\uu^{\top} \PP \uu}}} & \text{ if } t > \frac{\nu^2 \sqrt{\uu^{\top} \PP \uu}}{\alpha}.
\end{cases}
$$
\end{lemma}
 
\begin{lemma}
\label{lem:subexp-var}
Suppose $X$ is a $(\nu, \alpha)$ sub-exponential random variable with $\frac{2}{\nu^2} \leq \frac{1}{\alpha^2}$. Then 
\[
\Var(X) \leq \frac{3}{2}\nu^2.
\]
\end{lemma}
\begin{proof}
We will use the inequality $x^2 \leq C \br{e^x + e^{-x}}, \, \forall x$ and $C = 1.5/e.$ This gives 
\begin{align*}
     \text{Var}(\lambda X) & \leq  C \E \br{e^{\lambda X} + e^{- \lambda X}} \\
    & \leq 2C\exp(0.5 \lambda^2 \nu^2)
    , \text{ if } |\lambda| \leq \frac{1}{\alpha}.
\end{align*}
We therefore have
\begin{align*}
    \text{Var}(X) &\leq C \nu^2 \frac{\exp(0.5 \lambda^2 \nu^2)}{0.5 \lambda^2 \nu^2} \\
    &= eC \nu^2 \text{ when } 0.5 \lambda^2 \nu^2 = 1.
\end{align*}
We get the result with $C = 1.5/e.$
\end{proof}

\begin{lemma}
\label{lem:vector-var}
If $\ww$ is a is a mean zero $\br{\sigma, \PP}$ sub-Gaussian random vector, then $\Cov(\ww) \preceq \sigma^2 \PP$. If $\ww$ is a is a mean zero $\br{\nu, \alpha, \PP}$ sub-exponential random vector with $\frac{2}{\nu^2} \leq \frac{1}{\alpha^2}$, $\Cov(\ww) \preceq \frac{3}{2} \nu^2 \PP$.
\end{lemma}

\begin{proof}
For all unit vector $\uu$, we have
\[
\uu^{\top} \Cov(\ww) \uu = \uu^{\top} \E(\ww \ww^{\top}) \uu = \E[(\ww^{\top} \uu )^2] = \Var(\ww^{\top} \uu).
\]

By definition, if $\ww$ is a is a mean zero $\br{\sigma, \PP}$ sub-Gaussian random vector, $\ww^{\top} \uu$ is $\sigma^2 \uu^{\top} \PP \uu$ sub-Gaussian. Therefore, 
\[ \uu^{\top} \Cov(\ww) \uu = \Var(\ww^{\top} \uu) \leq \sigma^2 \uu^{\top} \PP \uu
\]
as desired. 

Similarly, if $\ww$ is a is a mean zero $\br{\nu, \alpha, \PP}$ sub-exponential random vector, $\ww^{\top} \uu$ is $(\nu \sqrt{\uu^{\top} \PP \uu}, \alpha \sqrt{\uu^{\top} \PP \uu})$ sub-exponential. By Lemma~\ref{lem:subexp-var}, 
\[ \uu^{\top} \Cov(\ww) \uu = \Var(\ww^{\top} \uu) \leq \frac{3}{2} \nu^2 \uu^{\top} \PP \uu.
\]
\end{proof}

\begin{lemma}
\label{lem:epsilon}
For all $i \in [n]$ 
\begin{enumerate}
    \item $\epsilon_i$ is sub-Gaussian with $\sigma = 1$, and
    \item $\epsilon_i$ is $(4\sqrt{Aq}, 2/B)$ sub-exponential for any $A,B > 0$ satisfying $e^x < 1 + x + Ax^2$ for $x < B$.
\end{enumerate}
\end{lemma}
\begin{proof}
For the first claim, note that a random variable bounded in $[a, b]$ is sub-Gaussian with $\sigma^2 = \frac{(b-a)^2}{4}$. 
To prove second claim, we have
\begin{align*}
\exp(\lambda \epsilon_i) &= \tilde{q} \exp(2 \lambda(1-\tilde{q})) + (1-\tilde{q}) \exp(-2 \lambda \tilde{q}) \\
&= \exp(-2 \lambda \tilde{q}) (1 + \tilde{q}(\exp(2 \lambda)-1) ) \\
& < \exp \br{-2 \lambda \tilde{q}} \exp \br{\tilde{q}(\exp(2 \lambda)-1) )} \\
& \leq \exp \br{\tilde{q} A (2 \lambda)^2},
\end{align*}
for $2\lambda < B.$ The last step follows from $e^x \leq 1+ x + Ax^2$ for $x <B.$
Recall that $\tilde{q} < 2q$, we have
$$
\exp(\lambda \epsilon_i) \leq \exp \br{ \frac{16qA\lambda^2}{2}} 
$$
for $\lambda < B/2$ as desired.
\end{proof}

Let $\PP_i, \, i\in [n]$ be orthogonal projection matrices onto the span of $\{\xx_1,..,\xx_i\},$ that is, 
\[\PP_i := \PP_{i-1} + \frac{(\II - \PP_{i-1}) \xx_i \xx_i^{\top} (\II - \PP_{i-1})}{\|(\II - \PP_{i-1})\xx_i \|^2},
\]
with $\PP_0 = 0.$ 
\begin{lemma}
\label{lem:main}
Suppose $A,B > 0$ satisfy $e^x < 1 + x + Ax^2$ for all $x < B$. Let $\sigma^2 :=c/2q$, $\nu^2 := 8Ac$ and $\alpha := 2/B$. Then 
\begin{enumerate}
    \item If $\ww_{i-1}$ is $(\sigma, \PP_{i-1})$ sub-Gaussian, $\ww_i$ is $(\sigma, \PP_i)$ sub-Gaussian.
    \item If $\ww_{i-1}$ is $(\nu, \alpha, \PP_{i-1})$ sub-exponential, $\ww_i$ is $(\nu, \alpha, \PP_i)$ sub-exponential.
\end{enumerate}
\end{lemma}

\begin{proof}
We have
\begin{align*}
    \E \left[ \exp( \lambda \ww_i^{\top} \uu) \right] &= \E \brs{ \E \brs{
    \exp( \lambda \ww_i^{\top} \uu) \big| \ww_{i-1} 
    } } \\
     &= \E \brs{ \E  \brs{ e^{ \lambda
    \left \langle \ww_{i-1}, 
\uu - 2 q  \frac{\xx_i^{\top} \uu}{c} \xx_i \right \rangle + \lambda \epsilon_i \xx_i^{\top} \uu} \big| \ww_{i-1}}} \\ 
& = \E \brs{ e^{ \lambda
     \langle \ww_{i-1}, 
\QQ_{\xx_i, c} \uu  \rangle} \E  \brs{ e^{\lambda \epsilon_i \xx_i^{\top} \uu} \big| \ww_{i-1}}}.
\end{align*}

First, suppose that $\ww_{i-1}$ is $(\sigma, \PP_{i-1})$ sub-Gaussian, we will prove that $\ww_i$ is $(\sigma, \PP_i)$ sub-Gaussian. By Lemma~\ref{lem:epsilon}, $\epsilon_i$ is 1-sub-Gaussian. Therefore, 
\begin{align*}
\E \brs{ e^{ \lambda
     \langle \ww_{i-1}, 
\QQ_{\xx_i, c} \uu  \rangle} \E  \brs{ e^{\lambda \epsilon_i \xx_i^{\top} \uu} | \ww_{i-1}}}
& \leq \E \brs{ e^{ \lambda
     \langle \ww_{i-1}, 
\QQ_{\xx_i, c} \uu \rangle} \cdot e^{\frac{1}{2} \lambda^2 \|\xx_i^{\top} \uu\|^2  }} \\
& \leq e^{ \frac{ \lambda^2  \sigma^2}{2} \br{ \QQ_{\xx_i, c} \uu}^{\top}  \PP_{i-1} \br{ \QQ_{\xx_i, c} \uu}  
+\frac{\lambda^2 }{2} (\uu^{\top} \xx_i \xx_i^{\top} \uu) }.
\end{align*}

We now consider the exponent (divided by the common factor $\lambda^2 $). It is sufficient to show 
\begin{align*}
     \frac{\sigma^2}{2} \br{ \QQ_{\xx_i, c} \uu}^{\top} \PP_{i-1} \br{ \QQ_{\xx_i, c} \uu}  
+ \frac{1}{2}u^{\top} \xx_i \xx_i^{\top} \uu  \leq \frac{\sigma^2}{2} \uu^{\top} \PP_i  \uu \, \, \forall \uu\\
 \iff   {\sigma^2} \QQ_{\xx_i, c} ^{\top}  \PP_{i-1} \QQ_{\xx_i, c}   
+ \xx_i \xx_i^{\top}  \preceq {\sigma^2} \PP_i.
\end{align*}

Using Proposition~\ref{prop:loewner} (with $L \leftarrow 1, C \leftarrow c/2q$) and assuming $ \sigma^2 = c/2q \geq 1$, we have
\begin{align*}
   \QQ_{\xx_i, c}  \PP_{i-1} \QQ_{\xx_i, c}  \sigma^2
+ \PP_{\xx_i}  &=  \frac{c}{2q}  \left( \II - \frac{2q}{c} \xx_i \xx_i^{\top} \right) \PP_{i-1} \left( \II - \frac{2q}{c} \xx_i \xx_i^{\top} \right) + \xx_i \xx_i^{\top} \\
&\preceq  \frac{c}{2q}  \PP_i.
\end{align*}

Therefore, $\ww_i$ is a $(c/2q, \PP_i) $ sub-Gaussian random vector.

Now suppose that $\ww_{i-1}$ is $(\nu, \alpha, \PP_{i-1})$ sub-exponential, we will prove that $\ww_i$ is $(\nu, \alpha, \PP_{i})$ sub-exponential. Again, by Lemma~\ref{lem:epsilon}, $\epsilon_i$ is $(4\sqrt{Aq}, 2/B)$ sub-exponential.  Therefore, 
\begin{align*}
\E \brs{ e^{ \lambda
     \langle \ww_{i-1}, 
\QQ_{\xx_i, c} \uu  \rangle} \E  \brs{ e^{\lambda \epsilon_i \xx_i^{\top} \uu} | \ww_{i-1}}}
& \leq \E \brs{ e^{ \lambda
     \langle \ww_{i-1}, 
\QQ_{\xx_i, c} \uu \rangle} \cdot e^{8Aq \lambda^2 \|\xx_i^{\top} \uu\|^2  }} 
\end{align*}
for $|\lambda| < \frac{2}{B|\xx_i^{\top} \uu|} = \frac{1}{\alpha \sqrt{\uu^{\top} \xx_i \xx_i^{\top} \uu}}$.
Since $\ww_{i-1}$ is $(\nu, \alpha, \PP_{i-1})$ sub-exponential,
\[
\E \brs{ e^{ \lambda
     \langle \ww_{i-1}, 
\QQ_{\xx_i, c} \uu  \rangle} } 
\leq 
e^{ \frac{ \lambda^2  \nu^2}{2} \br{ \QQ_{\xx_i, c} \uu}^{\top}  \PP_{i-1} \br{ \QQ_{\xx_i, c} \uu}} 
\]
for $|\lambda| < \frac{1}{\alpha \sqrt{ \uu^{\top} \PP_{i-1} \uu}}$.
Note that $\uu^{\top} \PP_{i} \uu$ is greater than both $\uu^{\top} \PP_{i-1} \uu$ and $\uu^{\top} \xx_i \xx_i^{\top} \uu$. We have
\begin{align*}
\E \brs{ e^{ \lambda
     \langle \ww_{i-1}, 
\QQ_{\xx_i, c} \uu  \rangle} \E  \brs{ e^{\lambda \epsilon_i \xx_i^{\top} \uu} | \ww_{i-1}}} 
\leq 
e^{ \frac{ \lambda^2  \nu^2}{2} \br{ \QQ_{\xx_i, c} \uu}^{\top}  \PP_{i-1} \br{ \QQ_{\xx_i, c} \uu}  
+\ {8Aq\lambda^2 } (\uu^{\top} \xx_i \xx_i^{\top} \uu) }
\end{align*}
for all $|\lambda| < \frac{1}{\alpha \sqrt{\uu^{\top} \PP_{i} \uu}}$.
Hence, it is sufficient to show 
\begin{align*}
     \frac{\nu^2}{2} \br{ \QQ_{\xx_i, c} \uu}^{\top} \PP_{i-1} \br{ \QQ_{\xx_i, c} \uu}  
+ 8Aq \uu^{\top} \xx_i \xx_i^{\top} \uu  \leq \frac{\nu^2}{2} \uu^{\top} \PP_i  \uu \, \, \forall \uu\\
 \iff   {\nu^2} \QQ_{\xx_i, c} ^{\top}  \PP_{i-1} \QQ_{\xx_i, c}   
+ 16Aq \xx_i \xx_i^{\top}  \preceq {\nu^2} \PP_i.
\end{align*}
This follows from Proposition~\ref{prop:loewner} by substituting $L \leftarrow 16Aq$ and $C \leftarrow c/2q$ and noting that $\nu^2 = 8Ac = Lc$.

\end{proof}

\begin{lemma}
\label{lem:success-prob}
If $c = \min(1/q, 9.3) \log(2n/\delta)$ then with probability at least $1-\delta$, we have 
\[
 |\ww_{i-1}^{\top} \xx_{i}| < c \text{ for all } i \in [n].
\]

\end{lemma}
\begin{proof}
By definition, $c$ is either equal to $\log(2n/\delta)/q$ or $9.3 \log(2n/\delta)$. We consider these two cases. 

First suppose $c = \log(2n/\delta)/q$. With $\sigma^2 = c/2q$, we have
$$
    c = \log(2n/\delta)/q = \sigma \sqrt{2 \log(2 n/ \delta)}.
$$
By Lemma~\ref{lem:subgaussian-bound}, 
\begin{align*}
 P\br{|\ww_{i-1}^{\top} \xx_{i}| > c} 
 &\leq P\br{|\ww_{i-1}^{\top} \xx_{i}| > c \sqrt{\xx_i^{\top} \PP_{i-1} \xx_i}} \\
 &\leq 2 \exp \br{ -\frac{c^2}{2 \sigma^2 }} \leq \delta/n.
\end{align*}
The result then follows by a union bound over $i \in [n].$

Now suppose $c = 9.3 \log(2n/\delta)$. Note that $A = 0.5803$ and $B = 0.4310$ satisfy $e^x < 1 + x + Ax^2$ for $x < B$. Let $\nu^2 =8Ac$ and $\alpha = 2/B$. We have 
$$
\frac{\nu^2}{\alpha} = \frac{8Ac }{2/B} = {4ABc} > c . 
$$
Therefore, by Lemma~\ref{lem:subexp-bound}, 
\begin{align*}
 P\br{|\ww_{i-1}^{\top} \xx_{i}| > c} 
 &\leq P\br{|\ww_{i-1}^{\top} \xx_{i}| > c \sqrt{\xx_i^{\top} \PP_{i-1} \xx_i}} \\
 &\leq 2 \exp \br{ -\frac{c^2}{2 \nu^2 }}
 = 2 \exp \br{-\frac{c}{16A}} \\
 &< 2 \exp \br{-\frac{c}{9.3}} \leq \delta/n.
\end{align*}
Again, the result follows by union bounding over $i \in [n].$
\end{proof}
Lemma~\ref{lem:success-prob} and Lemma~\ref{lem:main} together prove Theorem~\ref{thm:main}.

\subsection{Proof of Theorem~\ref{thm:multi}}

\begin{proof}[Proof of Theorem~\ref{thm:multi}]
Let $D(\delta)$ be the discrepancy obtained by \texttt{BinaryBalance} as a function of the failure probability $\delta$. We will show that $(2 \log k) D(\delta/k)$ is the corresponding discrepancy obtained by Algorithm~\ref{alg:multi}. 
Theorem~\ref{thm:multi} will then follow from Proposition~\ref{prop:balance}.
Notice that with probability $\delta/k$, each run of \texttt{BinaryBalance} at an internal node of $T$ has discrepancy $D(\delta/k)$. By union bounding over $O(k)$ internal nodes, we have that with probability $1-\delta$, all of the discrepancies are bounded by $D(\delta/k)$.

Assume all the discrepancies at the internal nodes in $T$ are bounded, we show how to bound the discrepancy between any two treatments. Let $l$ and $l'$ be any two leaves in $T$. The goal is to show that 
$$\left\| \frac{\alpha(l')}{\alpha(l') + \alpha(l)}s(l) - \frac{\alpha(l)}{\alpha(l') + \alpha(l)}s(l')\right\|_\infty$$ 
is small. 
First we relate $s(v)$ to $s(v_l)$ and $s(v_r)$ where $v_l, v_r$ are the left and right children of $v$. By definition,
\[ w(v) =  \frac{\alpha(v_r)}{\alpha(v_l) + \alpha(v_r)} s(v_l) - \frac{\alpha(v_l)}{\alpha(v_l) + \alpha(v_r)} s(v_r) \]
and 
\[s(v) = s(v_l) + s(v_r).\]
Therefore, 
\[w(v) = s(v_l) - \frac{\alpha(v_l)}{\alpha(v_l) + \alpha(v_r)} s(v),\]
and 
\[- w(v) = s(v_r) - \frac{\alpha(v_r)}{\alpha(v_l) + \alpha(v_r)} s(v).\]
Hence, both 
\[ \left\|s(v_l) - \frac{\alpha(v_l)}{\alpha(v)} s(v) \right\|_{\infty} 
\hspace{0.5cm} \text{and}  \hspace{0.5cm}
 \left\|s(v_r) - \frac{\alpha(v_r)}{\alpha(v)} s(v) \right\|_{\infty} \]
are bounded by $D(\delta/k)$.

Now consider $v_1, v_2$ and $v_3$ in $T$ such that $v_1$ is a child of $v_2$ and $v_2$ is a child of $v_3$. We have, by triangle inequality,  
\[ \left\|s(v_1) - \frac{\alpha(v_1)}{\alpha(v_3)} s(v_3) \right\|_{\infty} \leq  \left\|s(v_1) - \frac{\alpha(v_1)}{\alpha(v_2)} s(v_2) \right\|_{\infty} + \frac{\alpha(v_1)}{\alpha(v_2)}\left\|s(v_2) - \frac{\alpha(v_2)}{\alpha(v_3)} s(v_3) \right\|_{\infty} \leq \left(1+ \frac{\alpha(v_1)}{\alpha(v_2)}\right) D(\delta/k). \]
Let $l$ be a leaf in $T$ and let $l,v_1,v_2 \ldots r$ be the path from $l$ to the root $r$. Repeatedly applying the above relation along the path gives
\begin{equation}
    \label{eqn:sum-path}
    \left\|s(l) - \frac{\alpha(l)}{\alpha(r)} s(r) \right\|_{\infty} \leq 
    \left( 1 + \frac{\alpha(l)}{\alpha(v_1)} + \frac{\alpha(l)}{\alpha(v_2)} \ldots + \frac{\alpha(l)}{\alpha(r)} \right) D(\delta/k).
\end{equation}
Since there are at most $\log k$ nodes in the path from $l$ to $r$,
\[ \left\|s(l) - \frac{\alpha(l)}{\alpha(r)} s(r) \right\|_{\infty} \leq (\log k) D(\delta/k).\]
Finally, for any two leaves $l$ and $l'$,
\[ \left\|\frac{s(l)/\alpha(l) - s(l') /\alpha(l')}{1/\alpha(l) + 1/\alpha(l') }\right\|_{\infty} \leq \left\|\frac{s(l)/\alpha(l) - s(r) /\alpha(r)}{1/\alpha(l) + 1/\alpha(l') }\right\|_{\infty}  + \left\|\frac{s(l')/\alpha(l') - s(r) /\alpha(r)}{1/\alpha(l) + 1/\alpha(l') }\right\|_{\infty}  \leq (2 \log k) D(\delta/k).\]
\end{proof}

\begin{remark}
If all weights are uniform, the summation in (\ref{eqn:sum-path}) becomes a geometric series and can be bounded by a constant. Therefore, we can remove the factor $\log k$ in Theorem~\ref{thm:multi}.
\end{remark}

\subsection{Other Proofs}

\begin{proof}[Proof of Proposition~\ref{prop:z-dist}]
Let $\boldsymbol{B}=\left[\begin{array}{c}
\sqrt{\phi} \boldsymbol{I} \\
\sqrt{1-\phi} \boldsymbol{X}^{\top}
\end{array}\right]$. 
We have from Threorem~\ref{thm:main} that $\ww_n = \BB \eeta =  \left[\begin{array}{c}
\sqrt{\phi} \eeta \\
\sqrt{1-\phi} \boldsymbol{X}^{\top} \eeta
\end{array}\right] $ is a mean zero $(c/2q, \PP)$ sub-Gaussian random vector, where
\begin{align*}
 \PP
& =  \BB (\BB^{\top} \BB)^{-1} \BB^{\top} \\
& =  
\left[
\begin{array}{cc}
    \phi \QQ &  *\\
    *   & *
\end{array}
\right]. 
\end{align*}
Therefore, by sub-Gaussianity of $\ww_n,$ for any vector $\uu,$ we have 
\[
\E \brs{ \exp\br{\sqrt{\phi}\eeta^{\top} \uu }} \leq 
\exp \br{ \frac{c}{4q} \uu^{\top} \br{\phi \QQ} \uu.}
\]
We therefore have the sub-Gaussian claim. The sub-exponential result follows similarly.
\end{proof}

\begin{proof}[Proof of Proposition~\ref{prop:spectral-bound}]

Suppose $c = \log(2n/\delta)/q.$ We have from Proposition~\ref{prop:z-dist} that $\zz$ is a $(\sqrt{c/2q}, \QQ)$ sub-Gaussian vector.
This implies that 
$\Cov(\zz) \preceq \frac{c}{2q} \QQ.$ 

When $c= 9.3 \log(2n/\delta),$ we have from Proposition~\ref{prop:z-dist} that $\eeta$ is a $(\sqrt{8Ac}, \alpha, \QQ)$ sub-exponential vector. Now, Lemma~\ref{lem:subexp-var} gives that 
\begin{align*}
    \Cov(\zz) \preceq \frac{3}{2} 8Ac \QQ 
    = 12 A c \QQ.
\end{align*}
\end{proof}

\begin{proof}[Proof of Proposition~\ref{prop:balance}]
When $c = \frac{\log(4n/\delta)}{q}$ we have that with probability at least $1-\delta/2,$ $\BB \eeta$ is a $\br{\sqrt{c/2q}, \QQ}$ sub-Gaussian vector. Since $\BB \eeta = \left[\begin{array}{c}
\sqrt{\phi} \eeta \\
\sqrt{1-\phi} \boldsymbol{X}^{\top} \eeta
\end{array}\right]$ and $\QQ \preceq \phi \II,$ we have   
$\sqrt{1-\phi} \sum_i \eta_i \xx_i = \sqrt{1-\phi} \XX^{\top}\eeta$ is a $\br{\sqrt{1/\phi}\sqrt{c/2q}, \II } $ sub-Gaussian random vector. By sub-Gaussian concentration, we have with probability at least $1-\delta/2d$, $|\br{\XX^{\top} \eeta}^{\top} \ee_i| \leq \br{\sqrt{(c/2\phi (1-\phi) q}} \sqrt{4d/\delta}. $ The result follows by a union bound over $\ee_1,...,\ee_d$ and $\|\XX^{\top}\eeta \|_{2} \leq \sqrt{d}\|\XX^{\top}\eeta \|_{\infty}. $

When $c = 9.3 \log(4n /\delta),$ then with probability $1-\delta,$ $ \left[\begin{array}{c}
\sqrt{\phi} \eeta \\
\sqrt{1-\phi} \boldsymbol{X}^{\top} \eeta
\end{array}\right]$ is a $(\sqrt{8Ac}, \alpha, \QQ)$ random vector. Like before, this implies $\sqrt{1-\phi} \boldsymbol{X}^{\top} \eeta$ is a $(\sqrt{8Ac}, \alpha, \phi \II)$ random vector. By sub-exponential concentration, we have
\begin{align*}
P \br{ \left|\sqrt{1-\phi} \br{\XX^{\top} \eeta}^{\top} \ee_i \right| \geq t} 
\leq \left\{ \begin{array}{cc}
2 \exp \br{-t^2 \phi/2 \nu^2} & \text{ when } t \leq \frac{\nu^2}{\alpha \sqrt{\phi}} \\
2 \exp \br{ -t/2 \alpha} & \text{ otherwise }
\end{array} \right .
\end{align*}

Setting $t = \sqrt{ 2 \log (4d \delta) (8Ac)} \leq c \leq  \nu^2/\alpha, $ (when $n \geq d$), we get that with probability at least $1-\delta/2d,$ we have 
\[
|\sqrt{1-\phi} \br{\XX^{\top} \eeta}^{\top} \ee_i| \leq 9.3 \sqrt{ \frac{ \log (4d /\delta) \log (4n/\delta)}{\phi}.}
\]
The rest follows by a union bound.
\end{proof}

\begin{proof}[Proof of Proposition~\ref{prop:ate-concentration}.]
First consider the case when $c = \log(2n/\delta)/q = \sigma \sqrt{2 \log(2 n/ \delta)}$. By Proposition~\ref{prop:z-dist}, $\eeta$ is $(\sigma, \QQ)$ sub-Gaussian with $\sigma = \sqrt{c/2q}$. From Lemma~\ref{lem:subgaussian-bound}, we have
\begin{align*}
    P(|\eeta^{\top}\mmu| > c\sqrt{\mmu^T\QQ\mmu}) \leq 2 \exp \br{ - \frac{c^2}{2\sigma^2}} = \delta / n.
\end{align*}
Now consider $c = 9.3 \log(2n/ \delta)$. By Proposition~\ref{prop:z-dist}, $\eeta$ is $(\nu,\alpha, \QQ)$ sub-exponential with $\nu = \sqrt{8Ac}$. Note that
$$
{\nu^2}/{\alpha} = {8Ac } /{(2/B)} = {4ABc} > c . 
$$
From Lemma~\ref{lem:subexp-bound}, we have
\begin{align*}
    P(|\eeta^{\top}\mmu| > c\sqrt{\mmu^T\QQ\mmu}) \leq 2 \exp \br{ - \frac{c^2}{2\nu^2}} \\
    = 2 \exp \br{ - \frac{c}{16A}}\leq \delta / n.
\end{align*}
The result then follows by a union bound.
\end{proof}

\begin{proof}[Proof of Proposition~\ref{prop:ate-ridge}.]
This follows from Proposition~\ref{prop:spectral-bound} and the proof of Theorem 3 in \cite{harshaw2020balancing}.
\end{proof}

\subsection{Robustness}
\label{sec:robust}
Proposition~\ref{prop:spectral-bound} immediately gives a bound on
$\lambda_{\max}(\Cov(\zz))$ and hence bounds accidental bias.

\begin{remark}[Accidental Bias]
With probability $\geq 1-\delta$ the maximum eigenvalue of $\Cov(\zz)$ satisfies
\[
\lambda_{\max}\br{\Cov(\zz)} \leq  \frac{\sigma^2}{\phi}.
\]
\end{remark}

\section{Simulations}
\label{app:sims}

\subsection{Description}

\begin{table*}[h!]
\centering
\begin{tabular}{rccc}
    \toprule
    \textbf{DGP Name} & $\mathbf{X}$ & $\mathbf{y(0)}$ & $\mathbf{y(a)\; \mathrm{\bf s.t.}\; a \neq 0}$\\
    \midrule
    \textbf{QuickBlockDGP} & $X_{i,k} \sim \mathcal{U}(0,10), \forall k \in \{1,2\}$ &  $\prod_{k=1}^2 X_{k} + \epsilon$ & 1 + y(0) \\
    \midrule
    \textbf{LinearDGP} & $X_{i,k} = \epsilon_{k}, k \in \{1,\dots,4\}$& $\mathbf{X} \beta + \frac{1}{10}\epsilon_{y(0)}$ & 1 + $\mathbf{X} \beta + \frac{1}{10}\epsilon_{y(1)}$\\
    \midrule
    \textbf{LinearDriftDGP} & $X_{i,k} = \frac{i}{N} + \epsilon_{k}, k \in \{1,\dots,4\}$& $\mathbf{X} \beta + \frac{1}{10}\epsilon_{y(0)}$ & 1 + $\mathbf{X} \beta + \frac{1}{10}\epsilon_{y(1)}$\\
    \midrule
    \textbf{LinearSeasonDGP} & $X_{i,k} = \sin(2\pi\frac{i}{N}) + \epsilon_{k}, k \in \{1,\dots,4\}$& $\mathbf{X} \beta + \frac{1}{10}\epsilon_{y(0)}$ & 1 + $\mathbf{X} \beta + \frac{1}{10}\epsilon_{y(1)}$\\
    \midrule
    \textbf{QuadraticDGP} & $X_{i,k} = 2\beta_k - 1, k \in \{1, 2\}$& $\mu_0 + \frac{1}{10}\epsilon_{y(0)}$ & 1 + $\mu_0 + \frac{1}{10}\epsilon_{y(1)}$\\
    &$\mu_0 = X_1 - X_2 + X_1^2 + X_2^2 - 2 X_1 X_2$&& \\
    \midrule
    \textbf{CubicDGP} & $X_{i, k} = 2\beta_k - 1, k \in \{1, 2\}$& $\mu_0 + \frac{1}{10}\epsilon_{y(0)}$ & 1 + $\mu_0 + \frac{1}{10}\epsilon_{y(1)}$\\
    &$\mu_0 = X_1 - X_2 + X_1^2 + X_2^2 - 2 X_1 X_2$&& \\
    &$+ X_1^3 -X_2^3 - 3 X_1^2 X_2 + 3X_1 X_2^2$&&\\
    \midrule
    \textbf{SinusoidalDGP} & $X_{i,k} = 2\beta_k - 1, k \in \{1, 2\}$& $\mu_0 + \frac{1}{10}\epsilon_{y(0)}$ & 1 + $\mu_0 + \frac{1}{10}\epsilon_{y(1)}$\\
    &$\mu_0 = \sin\left(\frac{\pi}{3} + \frac{\pi X_1}{3} - \frac{2 \pi X_2}{3}\right)$&& \\
    &$-6 \sin(\frac{\pi X_1}{3} + \frac{\pi X_2}{4}) +6 \sin(\frac{\pi X_1}{3} + \frac{\pi X_2}{6})$&&\\
    \bottomrule
\end{tabular}
\caption{Data generating processes used in simulations. All $\epsilon$s indicate a standard normal variate and all $\beta$s indicate a standard uniform variate. $i$ indicates a unit's index. Covariate vectors are row-normalized to unit norm, except for the QuickBlock simulation which just normalized relative to the maximum row norm.}
\label{tab:dgps}
\end{table*}

\clearpage

\subsection{Binary Treatments}
\paragraph{Bias.} Figure~\ref{fig:bias} shows that none of the examined methods are biased (but that does not imply that they are robust \citep{efron1971forcing, harshaw2020balancing}).
\begin{figure}
    \centering
    \includegraphics[width=0.475\textwidth]{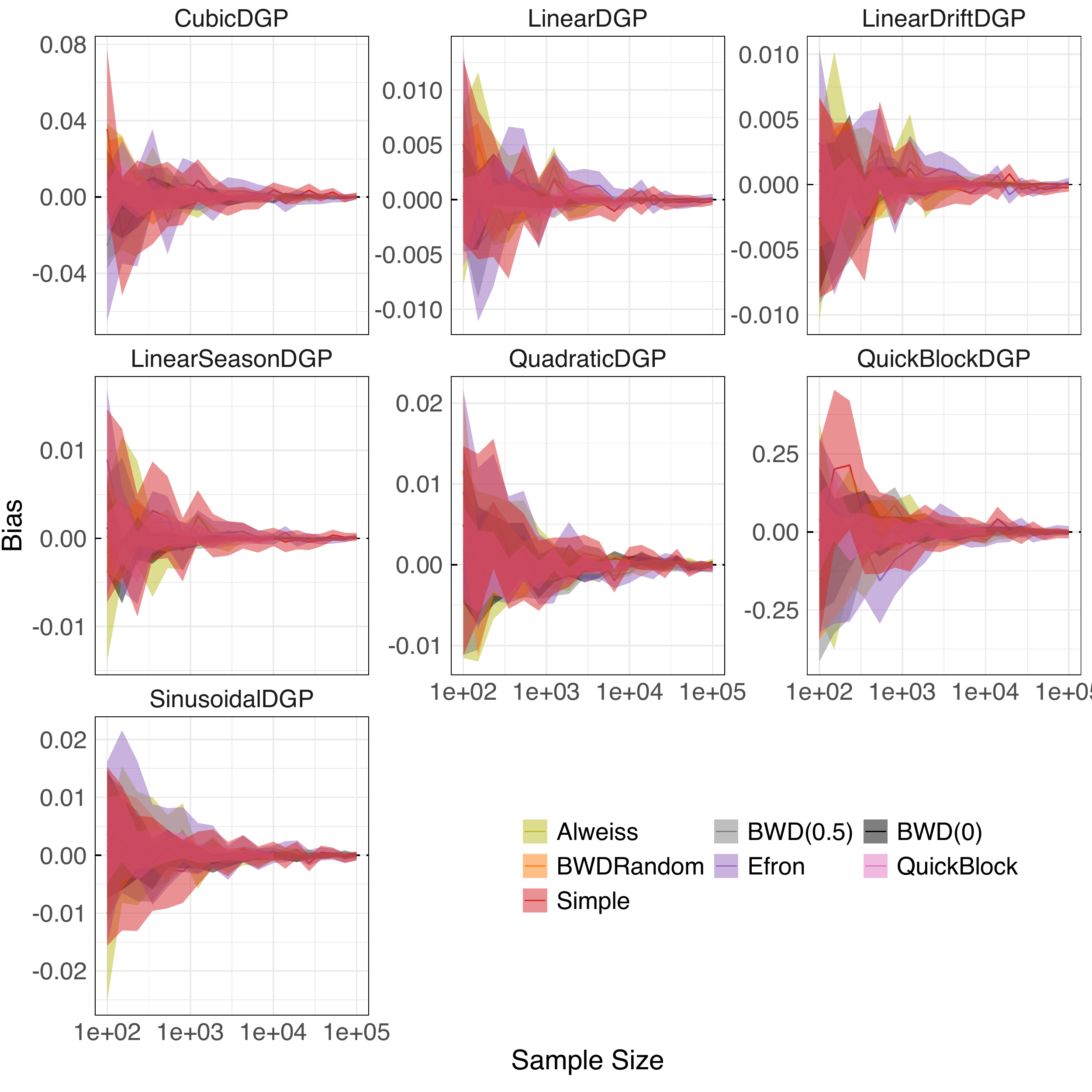}
    \caption{Bias}
    \label{fig:bias}
\end{figure}

\paragraph{Imbalance.} 
To measure linear imbalance, we calculate the $L_2$ norm of the difference in covariate means.
While this is far from the only measure of imbalance, it serves as an effective metric to demonstrate how well various metrics serve to eliminate linear imbalances, a common diagnostic used my experimenters.
We further normalize across sample-size by multiplying by the sample size.
Since all methods are unbiased, this has the effect of showing parametric convergence rates as a flat line in the graph.
Unsurprisingly, methods which directly optimize for linear imbalance perform very well in Figure~\ref{fig:imba}.
Our proposed algorithm has better finite sample imbalance minimization than does the algorithm of \citep{alweiss2020discrepancy} due to the finite sample improvements of \citep{dwivedi2021kernel}.

\begin{figure}
    \centering
    \includegraphics[width=0.45\textwidth]{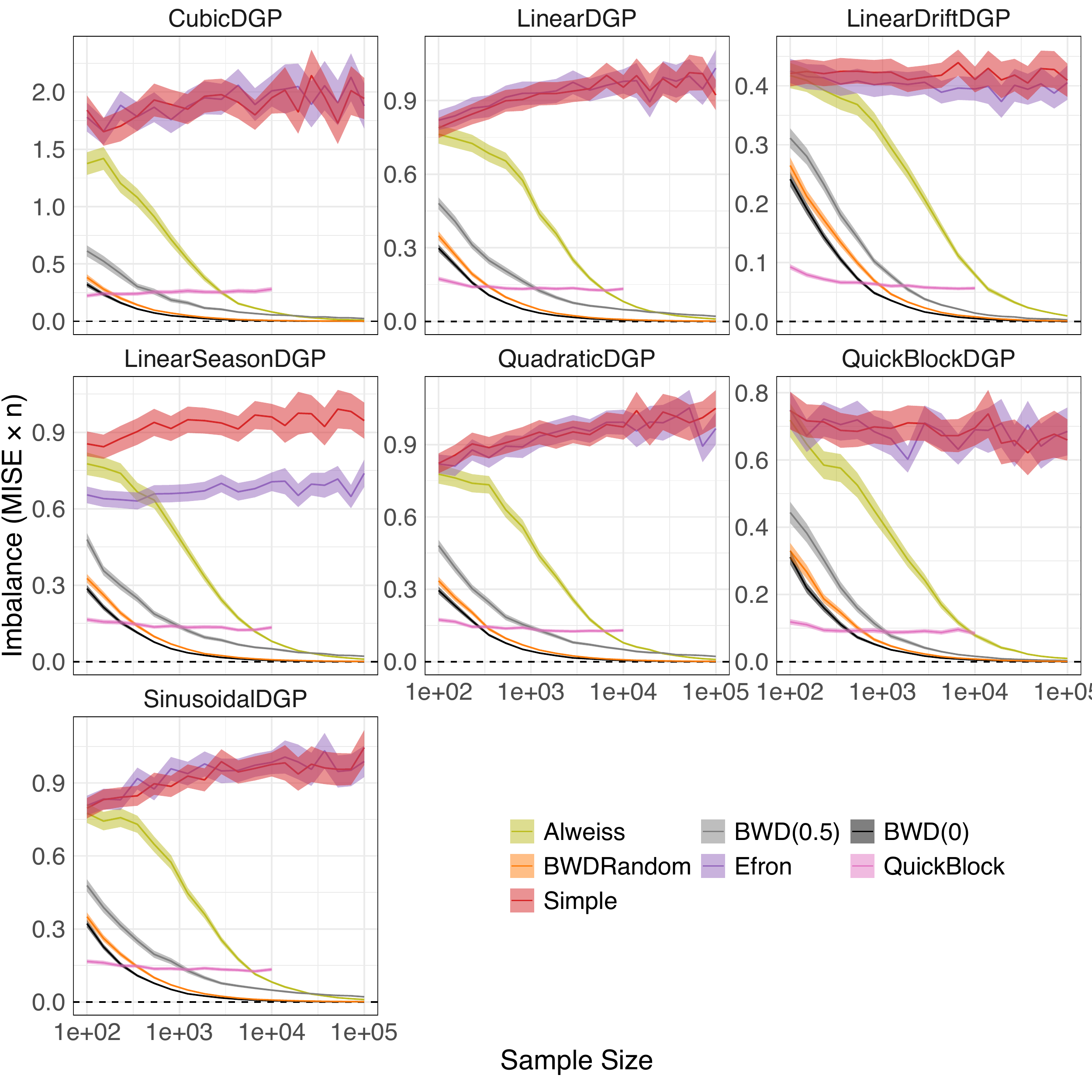}
    \caption{Imbalance. BWD is highly effective at eliminating linear imbalance between groups.}
    \label{fig:imba}
\end{figure}

\subsection{Non-uniform assignment}

\paragraph{MSE.}
Figure~\ref{fig:mse-q} shows the resulting mean squared-error attained by methods which support marginal probabilities not equal to one-half.
All methods perform well, with DM performing effectively on nearly linear processes, and QuickBlock performing slightly better when the true process is highly non-linear.

\begin{figure}
    \centering
    \includegraphics[width=0.475\textwidth]{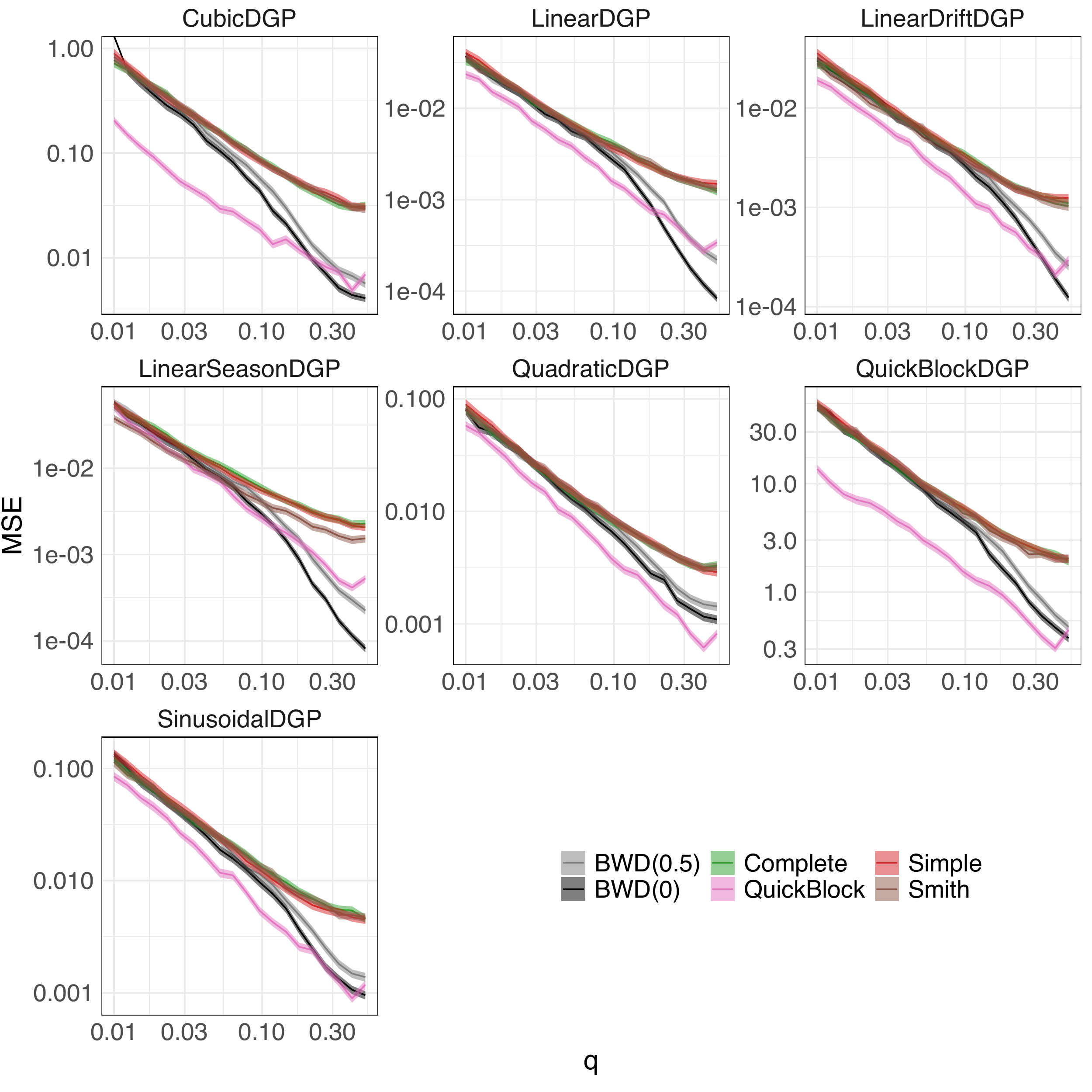}
    \caption{Marginal probability of treatment}
    \label{fig:mse-q}
\end{figure}

\paragraph{Marginal probability.}
The evaluation in Figure~\ref{fig:marginal-q} examines how closely each method hews to the desired marginal probability of treatment.
All methods do a good job of ensuring the appropriate marginal distribution. 

\begin{figure}
    \centering
    \includegraphics[width=0.475\textwidth]{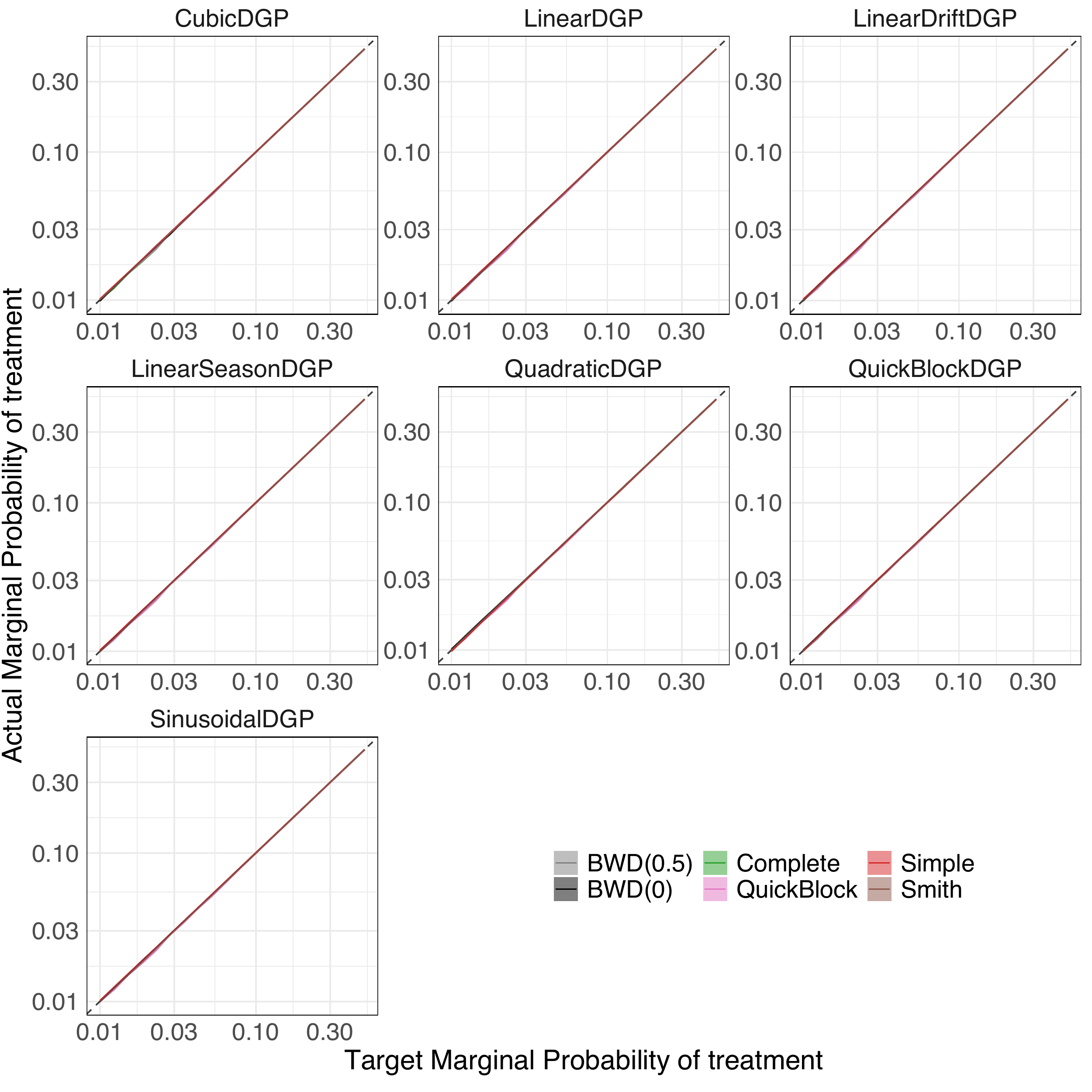}
    \caption{Marginal probability of treatment}
    \label{fig:marginal-q}
\end{figure}

\subsection{Multiple Treatments}

\paragraph{Entropy.}
To ensure that treatment is being assigned with the correct marginal probability, we can measure the normalized entropy of the empirical marginal treatment probabilities.
If the values are perfectly uniform, then the value will be exactly one.
As the normalized entropy decreases, the marginal probabilities are more uneven, indicating a failure to match the desired marginal distribution.
BWD performs very similarly to complete randomization (which almost exactly matches the desired marginal probabilities), slightly out-performing \citep{smith1984sequential}.
Note that while \citet{smith1984sequential} only seeks to optimize the marginal probability of treatment for each unit, BWD additionally balances covariates.

\begin{figure}
    \centering
    \includegraphics[width=0.475\textwidth]{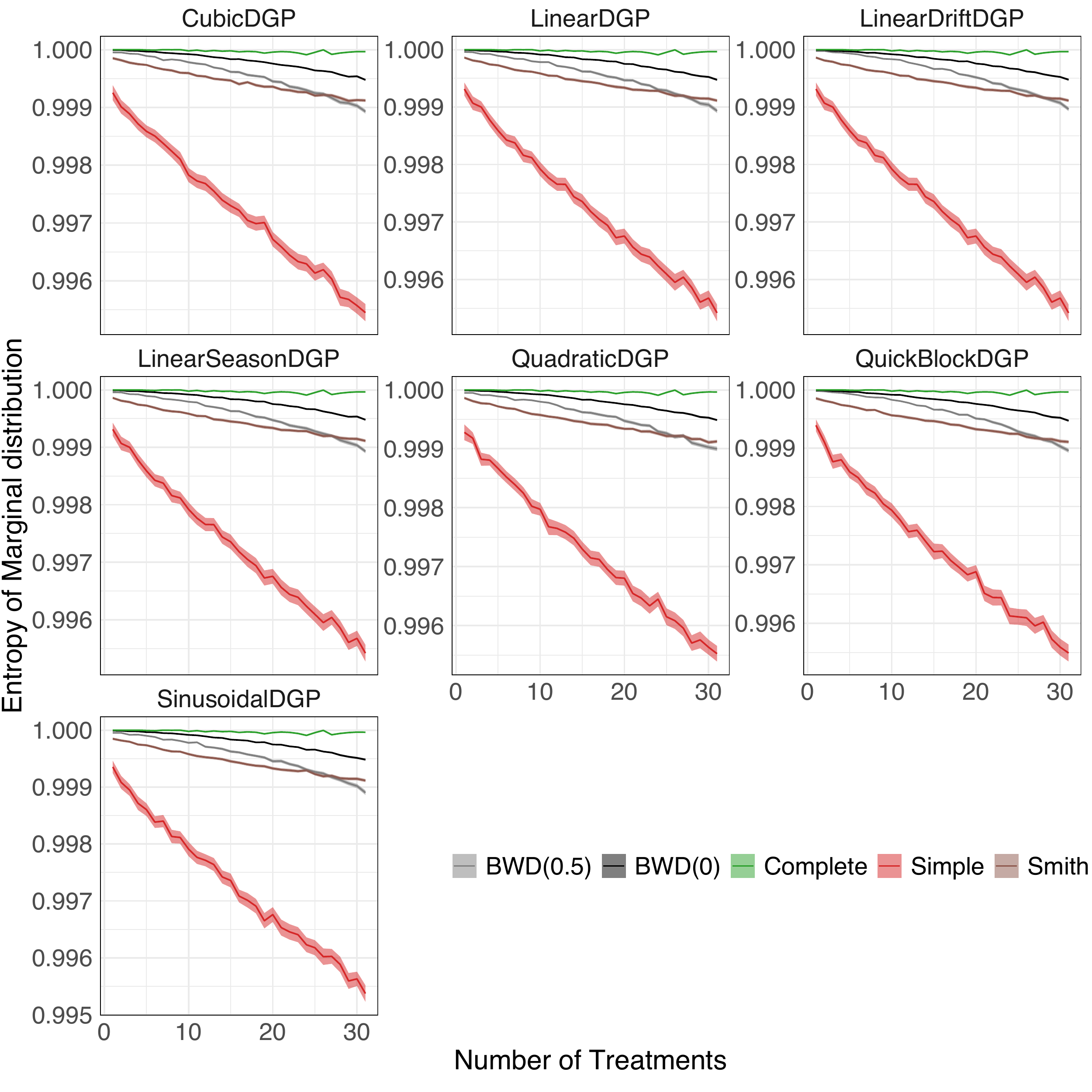}
    \caption{Entropy}
    \label{fig:entropy}
\end{figure}

\end{appendix}

\bibliographystyle{abbrvnat}

\end{document}